\documentclass[hidelinks]{siamart251216}

\usepackage{graphicx}
\usepackage{multirow}
\usepackage{amsmath,amssymb,amsfonts}
\usepackage{booktabs}
\usepackage{xcolor}
\usepackage{textcomp}
\usepackage[numbers,sort&compress]{natbib}

\newcommand{\R}{\mathbb{R}}

\newcommand{\singletonExcludedLGGz}{0.81}

\newsiamremark{remark}{Remark}
\newsiamremark{hypothesis}{Hypothesis}
\crefname{hypothesis}{Hypothesis}{Hypotheses}
\newsiamthm{claim}{Claim}
\newsiamremark{fact}{Fact}
\crefname{fact}{Fact}{Facts}

\makeatletter
\let\SM@origTheTheorem\thetheorem
\newcommand{\settheoremnum}[1]{\renewcommand{\thetheorem}{#1}}
\newcommand{\restoretheoremnum}{\let\thetheorem\SM@origTheTheorem}
\makeatother

\headers{A Null Model for Mapper Subtype Claims}{C. M. Topaz}

\title{A Null Model for Assessing\\ Mapper-Based Subtype Claims\thanks{Submitted to the editors \today.
\funding{This work received no external funding.}}}

\author{Chad M. Topaz\thanks{Williams College, Williamstown, MA, and University of Colorado Boulder, Boulder, CO
  (\email{cmt6@williams.edu}).}}

\ifpdf
\hypersetup{
  pdftitle={A Null Model for Assessing Mapper-Based Subtype Claims},
  pdfauthor={Chad M. Topaz}
}
\fi

\begin{document}

\maketitle

\begin{abstract}
The Mapper algorithm from topological data analysis constructs a graph
summarizing the shape of a high-dimensional dataset, and groups of data
points identified within this graph are widely interpreted as evidence
of distinct subtypes.  However, the covariance structure of the data
alone can make such groups appear differentiated, even when no subtypes
are present.  Existing validation approaches do not account for this effect and thus cannot distinguish covariance
artifacts from genuine subtypes.  We propose a Gaussian null model that generates reference data matching
the sample covariance matrix.  We pair it with a test statistic that measures mean-level
differentiation between communities.  In an idealized setting, we prove that covariance geometry alone causes
Mapper communities to differ in their average feature profiles, and we
show that a simpler label-permutation baseline cannot detect this effect.  Simulations confirm
well-controlled Type~I error under Gaussian data.  We apply the
framework to four published Mapper analyses spanning breast cancer gene
expression, Congressional voting, NBA player performance, and lower-grade
glioma genomics.  In every case, once outlier singleton communities are accounted for,
the observed differentiation does not exceed what the
null produces at the $\alpha = 0.05$ level.  This result does not rule out subtypes in these datasets, but it does indicate that the observed
structure is consistent with what covariance geometry alone can produce.  Stronger evidence would be needed to support a subtype claim.
\end{abstract}




\section{Introduction}\label{sec:intro}

The Mapper algorithm~\citep{Singh2007} from topological data analysis
(TDA) constructs a graph that summarizes the shape of a
high-dimensional dataset.  The algorithm projects the data onto a
low-dimensional function, divides the range of that function into
overlapping intervals, clusters the data points within each interval,
and connects clusters that share points.  The resulting graph can reveal
branching, loops, and densely connected groups of nodes called
communities, and these communities are routinely interpreted as evidence
of distinct subgroups or subtypes.  Since its introduction, Mapper has been applied to breast cancer
genomics~\citep{Nicolau2011, Lum2013}, diabetes stratification
\citep{Li2015, Wamil2023}, cancer subtyping~\citep{Rafique2020},
neuroimaging~\citep{Saggar2018, Geniesse2022, Hasegan2024}, single-cell
RNA-seq~\citep{Rizvi2017}, Congressional voting~\citep{Lum2012}, NBA player
performance~\citep{Lum2013}, and protein folding~\citep{Yao2009}, among
other domains.

When Mapper produces a graph with nontrivial structure, a natural
statistical question arises: does that structure reflect genuine
subgroups in the data, or is it consistent with what the algorithm would
produce from data with no subgroup structure but the same
\emph{covariance geometry}?  By covariance geometry we mean the pattern of variances, correlations, and
principal-component directions that characterize the data cloud's shape.
This question is important because Mapper requires the user to choose a
projection function, a partition granularity, and a clustering method,
and its output is sensitive to all of these
\citep{Dey2016, Alvarado2024, VejdemoJohansson2025}.  Yet in a survey of ten published Mapper studies
(Table~\ref{tab:null-audit-sm}), we find that null-model testing of graph structure is~rare.

To our knowledge, the only published null model for Mapper graph
features in cross-sectional data is that of~\cite{Lum2013}, which
compares Mapper output against data drawn from independent Gaussians
with common variance.  This null does not preserve correlation
structure, so it cannot distinguish community differentiation arising
from the covariance geometry of the data from differentiation arising
from additional structure such as distinct subpopulations.  The principle that null
models must preserve covariance geometry is well established in cluster
validation.  The gap statistic of~\cite{Tibshirani2001} generates
reference data aligned with the principal components of the observed
data, and SigClust~\citep{Liu2008} tests cluster significance against a
single Gaussian that matches the observed covariance structure.  This covariance-aware logic has not previously been applied to testing
Mapper graph structure.

We propose a \emph{covariance-preserving Gaussian null} for
Mapper-based subtype claims.  The idea is to generate synthetic data
that matches the observed covariance structure, apply the identical
Mapper pipeline, and ask whether the real data's community
differentiation exceeds what this benchmark produces.  The null model is
a general framework: any scalar summary of the Mapper graph can serve as
the test statistic.  We pair it here with a community-dissociation
metric that measures whether Mapper communities differ in their average
feature profiles (Section~\ref{sec:test-stat}), and we study this
combination's strengths and limitations in detail.

We prove that covariance geometry alone
drives positive population dissociation: even when data come from a
single multivariate Gaussian with no subtypes, filter-based partitions
produce communities with detectably different average feature values
(Theorem~\ref{thm:dissociation}).  This identifies the mechanism by which a null that does not preserve
covariance structure, such as the null of~\cite{Lum2013}, can yield
significance that reflects covariance geometry rather than subtype
structure.  One might ask
whether a simpler baseline---permuting community labels---would
suffice, but Theorem~\ref{thm:permutation} shows it does not: the permuted dissociation vanishes as sample size grows, so
the baseline always rejects, declaring spurious significance whether or not genuine subtypes exist.  The
structured null we propose asks
whether the observed community differentiation exceeds what a
covariance-matched Gaussian benchmark can generate.

We apply our framework to four case studies: three approximate
replications of analyses from~\cite{Lum2013} covering breast cancer gene
expression, U.S.\ Congressional voting records, and NBA player
performance, and a fourth using The Cancer Genome Atlas (TCGA) lower-grade glioma data from
\cite{RabadanCamara2020}.  In every case, the community differentiation
reported as evidence of subtypes does not exceed what the
covariance-preserving null produces at the $\alpha = 0.05$ level.  Two
analyses initially appear significant, but diagnostic examination
reveals that their rejections are driven by outlier communities rather
than systematic differentiation.  Non-rejection does not rule out the
existence of subtypes in these datasets, but it does indicate that the
observed community structure is consistent with what covariance geometry
alone can produce, and that stronger evidence would be needed to support
a subtype claim.  Conversely, rejection would mean that the community
differentiation exceeds what the data's variances and correlations can
explain under a Gaussian assumption, providing evidence for structure
beyond covariance geometry.

Section~\ref{sec:background} reviews the Mapper algorithm and the
statistical literature on cluster validation.
Section~\ref{sec:null} formalizes the covariance-preserving null, the
test statistic, and the testing procedure, and establishes the
theoretical properties of the test statistic.
Section~\ref{sec:simulation} presents a simulation study characterizing
Type~I error and sensitivity to non-Gaussian distributions.
Section~\ref{sec:results} applies the framework to four empirical case
studies, and Section~\ref{sec:discussion} discusses limitations,
extensions, and implications for practice.


\section{Background}\label{sec:background}

Mapper has been widely adopted for
subtype discovery across many domains.  We describe the algorithm here
and then explain why its outputs require null-model testing.

Mapper constructs a graph summarizing the shape of a dataset.  The key
input is a \emph{filter function} $f\colon \R^p \to \R$, which reduces
each high-dimensional data point to a single number.  The choice of
filter determines which aspect of the data the algorithm organizes
around.  A principal component projection, for instance, organizes points
by the direction of greatest variance.  A density estimate organizes
them by how crowded each point's neighborhood is, and an eccentricity
function organizes them by how central or peripheral each point
is~\citep{Singh2007, Lum2013}.  Given a point cloud
$X = \{x_1, \ldots, x_n\} \subset \R^p$ and a filter~$f$, the
algorithm proceeds in four steps:
\begin{enumerate}
\item[1.] \textbf{Cover the filter range.}  Choose a collection
  $\mathcal{U} = \{U_1, \ldots, U_N\}$ of $N$ overlapping intervals
  that together cover the range of~$f$ on the data.  The intervals
  typically have equal length and a fixed percentage of overlap, called
  the \emph{gain}.
\item[2.] \textbf{Pull back to the data.}  For each interval $U_i$,
  collect the \emph{preimage} $f^{-1}(U_i) \cap X$: the subset of data
  points whose filter values fall in~$U_i$.
\item[3.] \textbf{Cluster within each preimage.}  Apply a clustering
  algorithm independently to $f^{-1}(U_i) \cap X$.  Each resulting
  cluster becomes a \emph{vertex} (node) in the Mapper graph.
\item[4.] \textbf{Connect shared clusters.}  Draw an edge between two
  vertices whenever their corresponding clusters share at least one data
  point.  Shared points arise because adjacent intervals overlap, so a
  single data point can appear in multiple preimages.
\end{enumerate}
The output is an undirected graph whose connected components, branching
patterns, and loops encode topological features of the data as seen
through the filter~$f$.  Mapper can be viewed as a discrete,
finite-sample approximation of the Reeb graph of~$f$, a construction
from algebraic topology that collapses each level set of~$f$ to a
single point and retains only the connectivity between
them~\citep{Singh2007}.  Although we have described a scalar filter $f\colon \R^p \to \R$, the
filter can map to $\R^d$ for any $d \geq 1$; the cover then becomes a
product of one-dimensional covers and the construction is otherwise
identical.  In practice, one- and two-dimensional filters are most
common because higher-dimensional covers require exponentially more bins
and more data to populate them.  Three of our four case studies
(Section~\ref{sec:results}) use two-dimensional filters.

The construction depends on several additional user-specified choices.
The \emph{resolution}~$N$ and \emph{gain}~$g$ control the
granularity and overlap of the cover, and the \emph{clustering algorithm}
and its hyperparameters govern how each preimage is partitioned.  The
original implementation of~\cite{Singh2007} uses single-linkage
hierarchical clustering---a method that successively merges the two
closest clusters---and cuts the resulting tree (dendrogram) at the first
gap in a histogram of merge heights.  Alternative implementations use
DBSCAN, $k$-means, or agglomerative clustering with various linkages
\citep{vanVeen2019, Tauzin2021}.  Because of these many choices, the
Mapper graph can change substantially under small parameter perturbations
\citep{Dey2016, Chazal2021}.

Ref.~\cite{Alvarado2024} proves that
for any dataset and any target graph, there exist Mapper parameters that
produce it.  Ref.~\cite{VejdemoJohansson2025} shows that fixed-count
clustering introduces topological distortions.  Adaptive variants~\citep{Alvarado2025, Tao2025} reduce parameter
sensitivity, and under idealized conditions Mapper converges to the Reeb
graph~\citep{CarriereOudot2018, Brown2021, CarriereMichelOudot2018,
Dey2017}.  None of these developments, however, addresses whether the structure Mapper
finds exceeds what the algorithm would produce from a single continuous
population with the same covariance geometry.

Despite this sensitivity, nearly all published Mapper applications lack
rigorous null-model testing.
Table~\ref{tab:null-audit-sm} surveys thirteen analyses across ten
published Mapper studies and finds that only one~\citep{Lum2013}
includes any form of null-model comparison for graph structure.
Instead, researchers typically validate Mapper communities by checking
whether they correlate with external labels or domain knowledge.  In
biomedical applications, for instance, this might mean asking whether
the patients in one Mapper community share a clinical outcome or genetic
marker.  This approach can confirm that communities are meaningful, but
it does not address a more fundamental question: could the graph
structure itself arise from a single population with the same covariance
geometry?

Other areas of statistics have developed principled answers to this
kind of question.  In cluster validation, the gap
statistic~\citep{Tibshirani2001} and SigClust~\citep{Liu2008} both
compare observed cluster structure against null models that preserve the
data's distributional properties.  Within TDA, several frameworks test
whether topological features such as loops and connected components are
statistically significant rather than
artifacts of noise~\citep{Bobrowski2023, CarriereMichelOudot2018,
Chazal2021}; \cref{sm:validation} discusses these in more detail.
Separately, modularity optimization---the community-detection method we
use---has a known resolution limit that can prevent
detection of small communities~\citep{Fortunato2007}.

No analogous covariance-aware framework exists for Mapper.  The only
published null test for Mapper graph features is the
independence-based null of~\cite{Lum2013}, which compares Mapper output
to data drawn from independent Gaussians with equal variance.  As we
argue in Section~\ref{sec:null}, this null destroys all correlation
structure and is therefore too easy to beat.



\section{A structured null model for Mapper}\label{sec:null}

When Mapper produces a graph with nontrivial structure---multiple
connected components, branches, or densely connected
communities---a basic inferential question arises.  Does the observed
community differentiation exceed what the algorithm would produce from
data with no subgroup structure but the same covariance geometry?
Answering this question requires a \emph{null model}: a probability
distribution representing a single population with a specified
covariance structure.  The testing procedure is to generate many synthetic
datasets from the null, run the full Mapper pipeline on each one,
compute a summary statistic from each graph, and then ask whether the
real data's statistic is unusually large relative to the null
distribution.  The choice of null model matters.  A null that is too
easy to beat will declare every dataset significant, while a null that
is too difficult will obscure real discoveries.

The only null model used in the Mapper
literature~\citep{Lum2013} generates reference data by drawing points
independently from $\mathcal{N}(0, \sigma^2 I_p)$.  Here $I_p$ is the
$p \times p$ identity matrix, so every coordinate has the same variance
$\sigma^2$ and every pair of coordinates is uncorrelated.  The parameter
$\sigma^2$ is chosen to match the average marginal variance of the
observed data.  This null therefore preserves only the overall amount of
spread in the data, not how that spread is distributed across
directions; the resulting reference clouds are spherical.

In real data, however, the covariance structure already shapes the
geometry of the point cloud.  If variation is larger in some directions
than in others, or if variables move together, the cloud becomes
elongated, flattened, or tilted, even when the sample comes from a
single population with no subtypes.  Mapper responds to exactly this
geometry.  Consider an elongated cloud: data points near one end of the
elongation live in a different local neighborhood than points near the
other end.  When Mapper divides the filter range into overlapping
regions and clusters points within each region, these local differences
produce distinct clusters, which become separate nodes in the graph.
The result is a graph with nontrivial connectivity---not because the
data contain subgroups, but because the cloud's shape varies from one
end to the other.  A spherical null cloud, by contrast, looks roughly
the same in every direction, so its Mapper graphs, though they vary
from sample to sample due to finite-sample randomness, are simpler and
more uniform.  Comparing observed data to this null therefore
confounds covariance-driven graph complexity with any additional
structure in the data.  To ask whether community differentiation exceeds
what covariance geometry alone can produce, the null must preserve that
geometry.

\subsection{A structured null preserving covariance geometry}\label{sec:structured-null-def}

We propose to test Mapper output against a null model that preserves
the covariance structure of the data.  Given an observed data matrix
$X \in \R^{n \times p}$, where $n$ is the number of observations and
$p$ is the number of variables, define the sample covariance matrix
\begin{equation}\label{eq:sample-cov}
  \hat{\Sigma} = \frac{1}{n-1}(X - \bar{X})^\top(X - \bar{X}).
\end{equation}
The structured null generates reference data from
\begin{equation}\label{eq:structured-null}
  X^* \sim \mathcal{N}_p(0, \hat{\Sigma}).
\end{equation}
That is, each row of $X^*$ is drawn independently from a $p$-variate
normal distribution with mean zero and the same covariance matrix as
the observed data. The null hypothesis is that the data were generated
by a single multivariate normal distribution; the alternative is that
the data contain structure beyond what a single Gaussian with the
observed covariance can explain.

This null matches the observed covariance structure of the data: its
generating covariance has the same variances, correlations, and
principal-component directions as the sample covariance $\hat{\Sigma}$.
A reference dataset drawn from $\mathcal{N}_p(0, \hat{\Sigma})$ shares
the same population eigenvalue spectrum, the same blocks of correlated
variables, and the same effective dimensionality as the real
data---though any finite sample will exhibit sampling fluctuations
around these properties. The null reproduces the component of Mapper structure
generated by covariance geometry alone.  Community differentiation
that substantially exceeds the null therefore cannot be explained by a
single Gaussian with the observed covariance alone, and points toward
additional structure such as multiple modes, distinct subpopulations, or
non-Gaussian clustering.

Why choose the multivariate normal rather than some other distribution
with the same covariance?  Among all distributions with a given mean and
covariance, the multivariate normal has the highest
entropy~\citep{CoverThomas2006}, meaning it carries the least additional
structure beyond these two moments.  It is therefore the most
conservative choice: any non-Gaussian distribution with the same
covariance would introduce extra structure that could make the null
harder to beat.  This is the same logic underlying
SigClust~\citep{Liu2008}, applied here to Mapper graph structures.  The
null is straightforward to simulate via Cholesky or eigenvalue
decomposition of $\hat{\Sigma}$.

We can now state the testing procedure formally.  Let $T(\cdot)$
denote a scalar test statistic computed from a Mapper graph; we define
our choice of $T$ in Section~\ref{sec:test-stat}.  The conceptual null
hypothesis is:
\begin{align}
  H_0 &: X_1, \ldots, X_n \text{ are drawn from a $p$-variate} \notag \\
  &\quad \text{Gaussian with covariance matrix } \Sigma. \label{eq:H0}
\end{align}
After preprocessing, we center the data, so we work with the equivalent
mean-zero form $\mathcal{N}_p(0, \Sigma)$.  Because $\Sigma$ is
unknown, we implement the null as a \emph{plug-in benchmark test}: we
replace $\Sigma$ by the sample covariance matrix $\hat{\Sigma}$ computed
from the preprocessed data, generate null replicates from
$\mathcal{N}_p(0, \hat{\Sigma})$, and reapply the full analysis pipeline
to each replicate.  This means recomputing the distances, filter
function, Mapper graph, community detection, and test statistic each
time.

The inferential question is whether the observed test statistic is
unusually large relative to this covariance-matched Gaussian benchmark.
We define the $z$-score and Monte Carlo $p$-value formally in
Section~\ref{sec:mc-procedure} after introducing the test statistic.
The test is conditional: it conditions on the preprocessing pipeline,
the analyst's chosen Mapper parameters, and the estimated covariance
$\hat{\Sigma}$.  It asks whether the observed community differentiation
exceeds what the estimated Gaussian null can produce under those same
choices.

Two important caveats apply.  First, the null is \emph{potentially
miscalibrated under non-Gaussian data}: if the true data-generating
distribution is not Gaussian but is unimodal with the same covariance,
the Gaussian null could produce more or less Mapper structure than the
true unimodal distribution.  The result may be over- or under-rejection
even when the data come from a single population.  Second, the null is
\emph{specific to covariance structure}: it preserves the covariance
matrix but not higher-order features of the distribution (skewness,
kurtosis, tail behavior).  Departures from normality that do not involve multiple
subpopulations---such as heavy tails or skewness---could lead to false
rejections.  A rejection therefore indicates that community
differentiation exceeds what the Gaussian benchmark produces, not
necessarily that subtypes are present.  We return to these limitations
in Section~\ref{sec:discussion}.

When $p$ is large relative to $n$, the sample covariance matrix
$\hat{\Sigma}$ may be singular or ill-conditioned.  In particular,
after centering, $\hat{\Sigma}$ has rank at most $n-1$, so it is
singular whenever $p \ge n$.  A singular or nearly singular covariance
matrix has eigenvalues at or near zero and is not suitable for standard
Cholesky-based simulation.  We address the singularity via ridge regularization or
eigendecomposition in Section~\ref{sec:mc-procedure}.

\subsection{Test statistic: the dissociation metric}\label{sec:test-stat}

Mapper produces a graph, but in subtype-discovery applications the
scientific claim of interest is typically that distinct regions of that
graph correspond to distinct subpopulations in the data.  To study
this, we partition the Mapper graph into
\emph{communities}---groups of densely connected vertices---using the
algorithm described in Section~\ref{sec:community-detection}.  The
null-model test then asks whether community differentiation exceeds
what a covariance-matched Gaussian benchmark can produce.  This
requires a scalar test statistic that is large when communities are
well separated in feature space and small when they are not.  Our
statistic measures whether any two communities differ in their average
feature values, rather than whether the Mapper graph is topologically
complex.

Suppose the Mapper graph has been partitioned into $K$ communities
$C_1, \ldots, C_K$.  Because the cover intervals overlap, a single
data point may belong to more than one Mapper vertex, and those
vertices may fall in different communities.  We assign each data point
to exactly one community by a plurality vote over its vertices
(details in Section~\ref{sec:community-detection}).  We divide the $p$ features into two fixed, disjoint blocks $A$ and $B$
of roughly equal size.  We use two predetermined feature blocks rather than a single average
over all $p$ features so that the statistic does not rest entirely on
one global aggregation, while keeping each block large enough for
reliable averages.  Finer splits would provide more looks but with
noisier per-block averages; two blocks are a simple compromise.
Sections~\ref{sec:res-nki} and~\ref{sec:res-nba} confirm that results
are stable across 50 independent random splits.

For each community $C_k$ and each feature block, we average the
feature values across the community's members and then across the
features in the block:
\begin{equation}\label{eq:block-mean}
  \mu_{k,A} = \frac{1}{|A|}
  \sum_{j \in A} \biggl(\frac{1}{|C_k|} \sum_{i \in C_k} x_{ij}\biggr),
  \qquad
  \mu_{k,B} = \frac{1}{|B|}
  \sum_{j \in B} \biggl(\frac{1}{|C_k|} \sum_{i \in C_k} x_{ij}\biggr).
\end{equation}
Thus $\mu_{k,A}$ is a single number summarizing how high or low
community~$k$'s members score, on average, across the features in
block~$A$, and likewise for $\mu_{k,B}$.  We define the
\emph{dissociation metric}
\begin{equation}\label{eq:dissociation}
  D = \max_{1 \le k < \ell \le K} \max\bigl(
    |\mu_{k,A} - \mu_{\ell,A}|,\;
    |\mu_{k,B} - \mu_{\ell,B}|
  \bigr).
\end{equation}
When $D$ is large, at least one pair of communities differs
substantially in average feature level on at least one block.  When
$D$ is small, no pair of communities is strongly separated in mean
feature level, though the statistic may miss other forms of
differentiation (see below).  We set $D = 0$ if the graph has fewer
than two communities.

Taking the maximum over community pairs makes the statistic sensitive
to the presence of even one well-separated pair of communities.
Because the maximum can also be driven by a small or noisy community,
we examine singleton sensitivity as a diagnostic in
Section~\ref{sec:res-nba}.  Because $\mu_{k,A}$ averages over $|A|$
features, the metric is most sensitive to broad, coherent mean shifts.
If community~$k$ is elevated on half of block~$A$'s features and
depressed on the other half, $\mu_{k,A}$ may be near zero despite
substantial feature-level differences.  The statistic is therefore less
sensitive to sparse or sign-mixed subtype signals; we discuss
alternative statistics targeting feature-level maxima in
Section~\ref{sec:discussion}.  Alternative test statistics---such as the graph's modularity score,
the number of communities, or the first Betti number---could be
substituted within the same Monte Carlo framework to target different
aspects of Mapper output.

\subsection{Theoretical properties of the dissociation metric}\label{sec:theory}

We now analyze an idealized version of the dissociation metric to
isolate the covariance-driven mechanism that motivates the null.
Sections~\ref{sec:community-detection}--\ref{sec:mc-procedure} describe
the full implemented procedure.  The
following results establish that, in this idealized setting, the
dissociation metric is expected to be positive under a non-spherical
Gaussian even when the data contain no subtypes, and that a simpler
label-permutation baseline cannot detect this effect.  All
proofs are given in \cref{supp:proofs}.

We work with \emph{interval partitions}: partitions of the data
induced by dividing the filter range into non-overlapping intervals.
We discuss how these relate to actual Mapper communities below.

Let $X \sim \mathcal{N}_p(0, \Sigma)$ with $\Sigma$ positive definite.  Write the eigenvalues of $\Sigma$ in
decreasing order as $\lambda_1 \ge \lambda_2 \ge \cdots \ge \lambda_p > 0$,
and assume $\lambda_1 > \lambda_2$ so that the unit top eigenvector $u_1$ is
uniquely defined up to sign.  Let the filter be the projection onto the first
principal component, $f(x) = u_1^\top x$, which is the most common
filter choice in practice.  Let $A, B$ be a partition of
$\{1, \ldots, p\}$ into two nonempty blocks of variables.  Define the average loading of each block on~$u_1$:
\begin{equation}\label{eq:avg-loading}
  \bar{u}_{1,A} = \frac{1}{|A|} \sum_{j \in A} (u_1)_j, \qquad
  \bar{u}_{1,B} = \frac{1}{|B|} \sum_{j \in B} (u_1)_j.
\end{equation}

\begin{theorem}[Covariance-driven dissociation]\label{thm:dissociation}
  Let $\{I_1,\ldots,I_K\}$ be a finite interval partition of~$\R$
  with $K \ge 2$ and each interval having positive probability
  under~$f$.  Define the population block mean
  $m_{k,A} = \frac{1}{|A|}\sum_{j \in A} E[X_j \mid f \in I_k]$ and
  similarly $m_{k,B}$, and set
  \begin{equation}\label{eq:dpop}
    D_{\mathrm{pop}}
    \;=\;
    \max_{k < \ell}\;
    \max\bigl\{\,|m_{k,A} - m_{\ell,A}|,\;
                  |m_{k,B} - m_{\ell,B}|\,\bigr\}.
  \end{equation}
  If $\max\bigl(|\bar{u}_{1,A}|,\allowbreak |\bar{u}_{1,B}|\bigr) > 0$, then
  \begin{equation}\label{eq:thm-bound}
    D_{\mathrm{pop}} \;\ge\; \max\!\bigl(|\bar{u}_{1,A}|,\;
    |\bar{u}_{1,B}|\bigr) \cdot
    \bigl|E[f \mid f \in I_{\max}] - E[f \mid f \in I_{\min}]\bigr|
    \;>\; 0,
  \end{equation}
  where $I_{\min}$ and $I_{\max}$ are the leftmost and rightmost
  positive-probability intervals.
\end{theorem}

The key identity behind Theorem~\ref{thm:dissociation} is the
Gaussian regression formula
\begin{equation}\label{eq:gaussian-regression}
  E[X_j \mid f = c] = (u_1)_j \cdot c.
\end{equation}
Points with different filter values have different conditional means,
and averaging over a feature block produces a nonzero between-community
difference whenever that block has a nonzero average loading on~$u_1$.  The bound is driven by
two factors: the loading of the feature blocks on $u_1$, and the
separation between conditional means of the filter in the extreme
intervals.  Since $f \sim N(0, \lambda_1)$, larger $\lambda_1$ increases
filter variance and widens this separation.  The loading condition
$\max(|\bar{u}_{1,A}|, |\bar{u}_{1,B}|) > 0$ holds for all but a
measure-zero set of positive definite matrices~$\Sigma$
(\cref{supp:genericity}), so
Theorem~\ref{thm:dissociation} applies generically.

Theorem~\ref{thm:dissociation} identifies a covariance-driven
dissociation mechanism tied to anisotropic variation along a preferred
filter direction: a non-spherical Gaussian with no subtypes will
produce positive population dissociation under a PC1 filter partition.  A
spherical null does not encode that anisotropy, which helps explain why
comparison against a spherical null can yield spurious significance.
A seemingly plausible alternative would be to test community labels by
permutation: randomly reassign data points to communities of the same
sizes and check whether the observed dissociation is unusually large.
The next result shows that this baseline is uninformative.

\begin{theorem}[Label-permutation null]\label{thm:permutation}
  Suppose the number of communities~$K$ is fixed and community sizes
  satisfy $n_k / n \to \pi_k \in (0,1)$ as $n \to \infty$.  Conditional on any observed
  data for which the finite-population variances $S_A^2$ and $S_B^2$
  remain bounded, under a uniformly random reassignment of data points
  to communities, preserving the sizes~$n_k$, the permuted dissociation
  converges to zero in probability at rate $n^{-1/2}$.
  That is, as the sample size grows, the probability that
  $D^{\mathrm{perm}}$ exceeds any fixed threshold vanishes.
\end{theorem}

Together, Theorems~\ref{thm:dissociation} and~\ref{thm:permutation}
imply that the label-permutation baseline will pass for large~$n$:
the observed dissociation is bounded away from zero while the permuted
dissociation converges to zero, so in this idealized asymptotic regime
the permutation null passes whether or not genuine subtypes are present.  The permutation null confirms only that communities are spatially
coherent along the filter axis, a tautological consequence of how
Mapper constructs them, and does not test whether the observed
community structure exceeds what covariance geometry alone can produce.

Theorems~\ref{thm:dissociation} and~\ref{thm:permutation} are stated
for interval partitions defined by the filter.  Mapper communities are
not exactly interval partitions---they depend on clustering within
preimages and on graph community detection---but they approximately
respect filter order.  Under a Gaussian, the data decompose as
\begin{equation}\label{eq:decomposition}
  X = u_1 f + Z, \qquad E[X \mid C] = u_1 \cdot E[f \mid C] + E[Z \mid C],
\end{equation}
where $Z = X - u_1 f$ is the residual after removing the PC1
component, independent of~$f$ because uncorrelated Gaussian
components are independent, and $C$ denotes a data point's community
membership.  If community membership is primarily determined by the
filter value, the average residual within each community is small
and the interval-partition results carry over approximately.  We do
not claim a formal theorem for the full Mapper pipeline; rather, the
interval-partition analysis isolates the mechanism that our empirical
studies then assess in practice.

\subsection{Community detection}\label{sec:community-detection}

The dissociation metric requires a partition of the
Mapper graph into communities.
We use the Louvain algorithm~\citep{Blondel2008}, which
partitions a graph into densely connected groups by maximizing the
\emph{modularity score}~$Q$.  Modularity compares the density of
edges within each group to what would be expected in a random graph with
the same degree sequence.  Louvain is fast, widely used, and does not require the number of
communities to be specified in advance.  Its output can depend on
internal tie-breaking; for each null replicate, we fix a random seed
before the Louvain call so that stochasticity in community detection
does not confound the null comparison (see \cref{sm:replication-details}).

Each data point is assigned to a community as follows. A data point
may belong to multiple Mapper vertices, since cover intervals overlap,
and these vertices may belong to different communities in the
Louvain partition. We assign each data point by plurality vote to the community that
appears most frequently among its vertices; ties
are broken by the lowest community index.  This hard assignment
discards the soft membership information that the overlapping cover
provides.  Existing Mapper-based subtype analyses typically assign
points to subtypes informally, by visual inspection of the graph.
A formal assignment rule is necessary here because the test statistic
requires explicit community means~\eqref{eq:block-mean}.  The same
rule is applied to every null replicate, so the comparison remains
valid.  Points that do not
appear in any Mapper vertex are left unassigned and excluded from the
dissociation computation.

Modularity optimization has a known resolution
limit~\citep{Fortunato2007}: communities smaller than a scale that
depends on the total number of edges in the graph may go undetected.
Because this limit applies equally to the observed analysis and to
every null replicate, it does not introduce an obvious directional bias
into the hypothesis test, though we cannot rule out more subtle
interactions between graph size and resolution.  The test can detect
community structure only at the resolution scale of the Louvain
algorithm, and may miss finer-grained subgroups.

\subsection{Monte Carlo testing procedure}\label{sec:mc-procedure}

The complete testing procedure is as follows.

\begin{enumerate}
\item[1.] \textbf{Compute the observed test statistic.} Given the data
  matrix $X$, compute the distance matrix, the filter function(s), and
  the Mapper graph using the analyst's chosen parameters. Apply the
  Louvain algorithm to obtain communities. Assign each data point to a
  community and compute the observed dissociation metric $D_{\mathrm{obs}}$.
\item[2.] \textbf{Estimate the sample covariance.} Compute the sample
  covariance matrix $\hat{\Sigma}$ from $X$ (after any preprocessing
  such as centering and scaling). If $\hat{\Sigma}$ is singular or
  ill-conditioned, there are two options.  \emph{Ridge regularization}
  replaces $\hat{\Sigma}$ by $\hat{\Sigma} + \epsilon I_p$ for a small
  $\epsilon > 0$, shifting all eigenvalues upward and making the matrix
  positive definite; the choice of $\epsilon$ is described in
  \cref{supp:ridge}.  Alternatively, the
  \emph{eigendecomposition approach} decomposes $\hat{\Sigma}$ into its
  eigenvectors and eigenvalues, discards the zero-eigenvalue directions,
  and generates null samples only in the subspace where the data actually
  vary.  We use ridge regularization for the lower-grade glioma (LGG) analysis
  (Section~\ref{sec:res-lgg}), where $p = 4{,}500 \gg n = 534$, and the
  eigendecomposition approach for the NKI
  (Section~\ref{sec:res-nki}) and voting
  (Section~\ref{sec:res-voting}) analyses.
\item[3.] \textbf{Generate null replicates.} For $b = 1, \ldots, B$:
  \begin{enumerate}
  \item[(a)] Draw $X^{*(b)} \sim \mathcal{N}_p(0, \hat{\Sigma})$ with
    the same sample size $n$ as the observed data.
  \item[(b)] Compute the distance matrix for $X^{*(b)}$ using the same
    metric as the observed analysis.
  \item[(c)] Compute the filter function(s) for $X^{*(b)}$.  Filters
    that are derived from the data, such as $L^\infty$ centrality,
    principal components, or density, are recomputed from scratch on
    $X^{*(b)}$.  Filters based on external variables not included in
    the covariance model require different treatment.  For example, in
    the NKI breast cancer analysis (Section~\ref{sec:res-nki}), one
    filter dimension is a clinical survival indicator that is not part
    of the gene expression matrix.  Because the null cannot regenerate
    this variable, we reuse the observed survival values, keeping each
    value attached to its original data-point index, and add small
    independent jitter so that duplicated filter values do not
    artificially force points into the same Mapper interval.  This holds the external variable
    fixed and tests only whether the feature data contain structure
    beyond covariance geometry.
  \item[(d)] Run Mapper on $X^{*(b)}$ with identical parameters
    (resolution, gain, clustering algorithm and settings) to obtain a
    null Mapper graph.
  \item[(e)] Apply the Louvain algorithm to the null Mapper graph,
    assign data points to communities, and compute the null dissociation
    $D^{*(b)}$.
  \end{enumerate}
\item[4.] \textbf{Compute the $z$-score.} The standardized test
  statistic is
  \begin{equation}\label{eq:zscore}
    z = \frac{D_{\mathrm{obs}} - \bar{D}^*}{s_{D^*}},
  \end{equation}
  where $\bar{D}^* = \frac{1}{B} \sum_{b=1}^B D^{*(b)}$ and
  $s_{D^*}$ is the standard deviation of the null dissociation values.
  Large positive values of $z$ indicate that the observed community
  structure is more differentiated than what arises under the structured
  null.
\end{enumerate}

For the empirical analyses we use $B = 200$ null replicates, except
for LGG where we use $B = 100$ due to computational cost; the
simulations use $B = 50$.  With $B = 200$, the standard error of the
estimated null mean is $s_{D^*}/\sqrt{200}$, providing adequate
precision for distinguishing significant from non-significant results.  The Monte
Carlo $p$-value $\hat{p}$ is valid regardless of the null distribution
shape.  When a null replicate produces a Mapper graph with fewer than
two communities, we set $D^{*(b)} = 0$.  In all of our empirical
and simulation analyses, every null replicate produces at least two
communities, so this convention is never invoked.

As a secondary diagnostic, we also compute the label-permutation null
analyzed in Theorem~\ref{thm:permutation}: we keep the data and
community sizes fixed, randomly reassign data points to communities,
and recompute~$D$.  As established above, this test confirms spatial
coherence but cannot distinguish covariance-driven coherence from
genuine subtypes.  We report it alongside the structured null to make
this contrast visible.


\section{Simulation study}\label{sec:simulation}

Before applying the structured null to real data, we characterize its
operating properties with known data-generating processes. The goals are
threefold.  First, we verify that the test controls Type~I error when data are
drawn from a single multivariate normal
(Section~\ref{sec:sim-typeI}).  Second, we quantify how much false
rejection rates inflate under non-Gaussian unimodal structure, heavy
tails, and skewness, where the Gaussian null is misspecified
(Section~\ref{sec:sim-typeI}).  Third, we examine the test's behavior against Gaussian mixtures with known
separation, which reveals what the test
actually measures (Section~\ref{sec:sim-power}).

Each scenario generates $R = 200$ independent datasets of $n = 300$
data points and applies the full structured null test with $B = 50$
null replicates per dataset.  The Mapper pipeline uses two-dimensional
equalized principal coordinates analysis (PCoA) filters, resolution~15,
gain~2.0, and 10 histogram bins.  We reject when $z > 1.645$,
the one-sided $\alpha = 0.05$ threshold, because only unusually
large dissociation counts as evidence against the null.

The data-generating processes span three regimes.  For Type~I error
calibration, we draw data from spherical and correlated Gaussians in
$p = 10$, $50$, and $500$ dimensions.  For misspecification, we use a
multivariate $t$ distribution with $\mathrm{df} = 5$ and a skewed
unimodal distribution.  For power, we use all-feature and sparse
mean-shift mixtures with separation $\delta \in \{0.5, 1.0, 2.0\}$
and heterogeneous-covariance mixtures with correlation $\rho \in
\{0.5, 0.8, 0.9\}$.  Full specifications are in \cref{sm:dgp}.

\subsection{Calibration under Gaussian nulls and sensitivity to misspecification}\label{sec:sim-typeI}

Table~\ref{tab:sim-typeI} reports the rejection rate at $\alpha = 0.05$
for each null scenario.  Under all Gaussian data-generating
processes---spherical ($p = 10$), correlated ($p = 10$, $50$, and
$500$)---the rejection rates range from 4.0\% to 6.0\%, all within one
Monte Carlo standard error of the nominal 5\% level
($\sqrt{0.05 \cdot 0.95 / 200} \approx 1.5\%$).  The $p = 500$
scenario uses ridge regularization and directly validates the test in
the $p > n$ regime relevant to the genomic case studies.  Doubling the
number of null replicates from $B = 50$ to $B = 100$ does not change
the rejection rate, so $B = 50$ provides adequate calibration.

Because $D$ is a maximum over non-negative quantities, its null
distribution is typically right-skewed, and the $z$-score threshold
$z > 1.645$ may be somewhat conservative.  The Monte Carlo $p$-value
$\hat{p}$, which makes no distributional assumption, is more reliable;
we report both throughout.

Under non-Gaussian unimodal distributions, the Gaussian null is
misspecified and overrejects.  The multivariate~$t$ with
$\mathrm{df} = 5$ produces a rejection rate of 79.5\%, because heavy
tails generate extreme observations that Mapper resolves into
communities with differentiated feature profiles, yielding dissociation
values the Gaussian null cannot match.  Skewness has a milder but
still meaningful effect, with a rejection rate of 23.0\%.  Asymmetric
marginals create directional density gradients that Mapper resolves
into communities with systematically shifted means.  Practitioners applying the test to
data with non-Gaussian marginals should interpret rejections
cautiously, as the test may detect departures from Gaussianity rather
than genuine subtype structure.  We return to this point in
Section~\ref{sec:discussion}.

\begin{table}[t]
\caption{Type~I error calibration and misspecification sensitivity.
The test is well-calibrated under Gaussian data but can severely
overreject under heavy tails or skewness.
We generate $R$ independent datasets of $n = 300$ points in $p$
dimensions from each distribution, apply the structured null test with
$B$ null replicates per dataset, and reject when $z > 1.645$
(one-sided $\alpha = 0.05$).  The first five rows test calibration
under correctly specified Gaussian nulls; all rejection rates fall
within one Monte Carlo standard error of the nominal 5\% level
($\sqrt{0.05 \cdot 0.95 / R} \approx 1.5\%$).  The last two rows
show overrejection when the Gaussian null is misspecified for
non-Gaussian unimodal data.}
\label{tab:sim-typeI}
\begin{tabular}{@{}lccccc@{}}
\toprule
\textbf{Distribution} & $p$ & $B$ & $R$ & \textbf{Mean} $z$ & \textbf{Rejection rate} \\
\midrule
Spherical Gaussian      & 10  & 50  & 200 & $-0.13$ & $0.045$ \\
Correlated Gaussian     & 10  & 50  & 200 & $\phantom{-}0.00$  & $0.060$ \\
Correlated Gaussian     & 50  & 50  & 200 & $-0.00$ & $0.040$ \\
Correlated Gaussian     & 500 & 50  & 200 & $\phantom{-}0.00$ & $0.050$ \\
Correlated Gaussian     & 10  & 100 & 200 & $\phantom{-}0.00$ & $0.050$ \\
\midrule
Multivariate $t$ ($\mathrm{df}=5$) & 10  & 50  & 200 & $\phantom{-}4.08$  & $0.795$ \\
Skewed unimodal         & 10  & 50  & 200 & $\phantom{-}0.59$  & $0.230$ \\
\bottomrule
\end{tabular}
\end{table}

\subsection{Behavior under Gaussian mixture alternatives}\label{sec:sim-power}

Table~\ref{tab:sim-power} reports results for all mixture
scenarios. For the all-feature mean-shift mixtures, the test does not reject
the structured null for \emph{any} separation, including the
largest ($\delta = 2.0$), where the mean $z$-scores are in fact
strongly \emph{negative} ($-2.52$ at $p = 10$, $-2.73$ at $p = 50$).
Rather than indicating a failure of the test, these results illuminate
what the structured null test actually measures.

\begin{table}[t]
\caption{Behavior of the structured null test under Gaussian mixture
alternatives.  The test does not reliably detect any of these mixtures,
because the mixture structure is already captured by the sample
covariance matrix, so the Gaussian null generates data that look
similar.  For each mixture type, we generate $R = 200$ datasets of
$n = 300$ points in $p$ dimensions and test with $B = 50$ null
replicates, rejecting when $z > 1.645$ (one-sided $\alpha = 0.05$).
Rows~1--4: all-feature mean-shift mixtures where every feature shifts by
$\pm\delta/2$.  Rows~5--6: sparse mixtures where only $k = 5$ of
$50$ features shift.  Rows~7--12: mixtures with no mean shift but
different covariance structures in each component, where $\rho$ is the
correlation among features within each component.}
\label{tab:sim-power}
\begin{tabular}{@{}lcccccc@{}}
\toprule
\textbf{Mixture type} & $p$ & $\delta$ & $\rho$ & $R$ & \textbf{Mean} $z$ & \textbf{Rejection rate} \\
\midrule
All-feature mixture        & 10 & 0.5 & 0.5 & 200 & $-0.15$ & $0.030$ \\
All-feature mixture        & 10 & 1.0 & 0.5 & 200 & $-0.85$ & $0.005$ \\
All-feature mixture        & 10 & 2.0 & 0.5 & 200 & $-2.52$ & $0.000$ \\
All-feature mixture        & 50 & 2.0 & 0.5 & 200 & $-2.73$ & $0.000$ \\
\midrule
Sparse ($k=5$)    & 50 & 1.0 & 0.5 & 200 & $-0.02$ & $0.035$ \\
Sparse ($k=5$)    & 50 & 2.0 & 0.5 & 200 & $-0.06$ & $0.065$ \\
\midrule
Hetero-covariance      & 10 & ---  & 0.5 & 200 & $\phantom{-}0.22$ & $0.110$ \\
Hetero-covariance      & 50 & ---  & 0.5 & 200 & $-0.20$ & $0.050$ \\
Hetero-covariance      & 10 & ---  & 0.8 & 200 & $\phantom{-}0.11$ & $0.070$ \\
Hetero-covariance      & 50 & ---  & 0.8 & 200 & $-0.39$ & $0.030$ \\
Hetero-covariance      & 10 & ---  & 0.9 & 200 & $\phantom{-}0.03$ & $0.065$ \\
Hetero-covariance      & 50 & ---  & 0.9 & 200 & $-0.40$ & $0.025$ \\
\bottomrule
\end{tabular}
\end{table}

The explanation lies in how the sample covariance relates to the
mixture structure.  The all-feature mixture shifts every feature by the same
amount, so the two component means $\pm(\delta/2)\mathbf{1}$ lie along
the all-ones direction.  After centering and scaling, the sample
covariance $\hat\Sigma$ absorbs this between-component variance: it
``sees'' the mixture as additional spread along one axis.  Draws from
$\mathcal{N}_p(0, \hat\Sigma)$ therefore produce an elongated data
cloud with geometry similar to the mixture, and Mapper applied to
these null draws generates community structure with comparable---or
even larger---dissociation values.  In short, the covariance null
absorbs the mixture structure because that structure is encoded in the
second-order statistics.  At larger separations the between-component
variance dominates $\hat\Sigma$ even more, hence the increasingly
negative $z$-scores as $\delta$ grows.

The sparse mixtures (rows 5--6 of Table~\ref{tab:sim-power}),
where only 5 of 50 features shift, tell the same story: rejection
rates of 3.5\% ($\delta = 1.0$) and 6.5\% ($\delta = 2.0$) are
indistinguishable from the nominal level.  Even when the mean shift is
concentrated on a small subset of features, the corresponding entries of
$\hat\Sigma$ capture the between-component variance, and the covariance
null absorbs the mixture signal.  This extends the scope statement
beyond uniform shifts: the structured null does not detect \emph{any}
mean-shift mixture whose signal is encoded in the sample covariance,
regardless of how many features carry the shift.

The heterogeneous-covariance mixtures (rows 7--12) probe whether the
test can detect alternatives that differ from the null only in
covariance structure.  The two components share a common mean (zero) but
have different within-component correlation structures, so the pooled
sample covariance is a compromise between the two component covariances.
A unimodal Gaussian drawn from this compromise should, in principle,
differ from the two-component mixture.  We test three correlation strengths
($\rho = 0.5, 0.8, 0.9$) at two dimensionalities ($p = 10, 50$).
Across all six scenarios, rejection rates range from 2.5\% to 11.0\%
and show no increasing trend with $\rho$; even at $\rho = 0.9$, the
test does not detect the heterogeneous-covariance signal.

The reason is that the dissociation metric measures differences in
community \emph{means}, and both components of the heterogeneous-covariance
mixture have the same mean vector.  The components differ only in
their within-component correlation patterns, but Mapper communities
formed from the mixture will not exhibit systematically different mean
profiles.  The null model could in principle generate different graph
topology than the mixture, but the dissociation metric is blind to this
because it summarizes community structure through average feature
values, not through covariance or higher-order moments.

This result clarifies the scope of the structured null test.  The test
does not detect subtype structure that is fully encoded in the sample
covariance---including both all-feature and sparse mean-shift
mixtures---and it does not detect mixtures whose components differ only
in their covariance structure when the test statistic is mean-based.
The structured null test is therefore \emph{not} a general-purpose
test for mixture structure or clustering.  Researchers seeking a general test for multimodality should consider
purpose-built alternatives such as the dip
test~\citep{Hartigan1985}, Silverman's kernel-bandwidth
test~\citep{Silverman1981}, or the SigClust procedure for
high-dimensional data~\citep{Liu2008}.  The structured null test
answers a narrower but practically important question: does the Mapper
graph reveal community differentiation that exceeds what a
covariance-matched Gaussian benchmark produces under the same pipeline?
A covariance-sensitive test statistic, such as one comparing
within-community covariance matrices across communities, could in
principle detect the heterogeneous-covariance signal, but we leave this
to future work.


\section{Empirical case studies}\label{sec:results}

We apply the structured null model of Section~\ref{sec:null} to four
Mapper-based analyses.  Three come from~\cite{Lum2013}---breast cancer
gene expression, U.S.\ Congressional voting, and NBA player
performance---and the fourth is a cancer genomics analysis
from~\cite{RabadanCamara2020}.  The~\cite{Lum2013} analyses were
originally conducted in Ayasdi, a proprietary Mapper platform that
was acquired by SymphonyAI in 2019 and is no longer available as a
standalone product.  Its internal parameters do not map directly to
open-source implementations.  We do not claim to reproduce the original results exactly.  Instead,
we build open-source Mapper pipelines and tune their parameters until
the resulting graphs are qualitatively similar to those
in~\cite{Lum2013} (see \cref{sm:ayasdi} for details).  Because the
structured null uses the same pipeline parameters for both observed
and null data, imperfect replication does not bias the test: the
comparison is always between the real data and null data processed
identically.

The fourth analysis, lower-grade glioma from~\cite{RabadanCamara2020},
uses a methodologically distinct pipeline and a dataset with
well-established molecular subtypes, providing an independent check on
the framework's generalizability.

For each analysis, we describe the data and Mapper pipeline, then
report the observed dissociation metric~$D_{\mathrm{obs}}$, the
$z$-score under the label-permutation null, and the $z$-score under
the structured null.  Table~\ref{tab:results} summarizes all results.
We also report empirical one-sided Monte Carlo $p$-values,
\begin{equation}\label{eq:mc-pvalue}
  \hat{p} = \frac{1 + \#\{D^{*(b)} \ge D_{\mathrm{obs}}\}}{1 + B},
\end{equation}
where $D^{*(b)}$ is the dissociation of the $b$th null replicate and
$B$ is the total number of replicates.  The $+1$ terms in numerator
and denominator follow~\cite{Davison1997}.  They ensure that
$\hat{p}$ is never exactly zero and is uniformly distributed under
the null.
A complete replication package, including all code, intermediate
results, and instructions for obtaining the data, is available at
\texttt{[URL to be added upon publication]}.

\begin{table}[t]
\centering
\footnotesize
\caption{Summary of empirical results.  No analysis exceeds the
covariance-preserving null at $\alpha = 0.05$; the two initial
rejections (NBA all, LGG all) are driven by singleton communities and
disappear when singletons are excluded.
NKI refers to the Netherlands Cancer Institute breast cancer gene
expression dataset; LGG refers to lower-grade glioma gene expression
data; ``multi'' excludes singleton communities from the dissociation
computation.  Source gives the original study replicated; Sec.\ gives
the subsection of this paper containing the full analysis.
For each analysis, we report the number of Louvain communities~$K$,
observed dissociation~$D_{\mathrm{obs}}$, the label-permutation
$z$-score~$z_{\mathrm{perm}}$, the structured-null
$z$-score~$z_{\mathrm{str}}$, and the empirical one-sided Monte Carlo
$p$-value~$\hat{p}_{\mathrm{str}}$ computed from $B$ null replicates
via~\eqref{eq:mc-pvalue}.  The NBA analysis uses two Mapper parameter settings; rows here
correspond to the first (resolution~30, gain~2.0), with consistent
results for the second reported in Section~\ref{sec:res-nba}.
For the NBA rows, the label-permutation null was recomputed after
singleton exclusion; the nearly identical $z_{\mathrm{perm}}$ values
(3.30 vs.\ 3.29) reflect the dominance of the large communities in
the permutation variance.  For LGG (multi), the label-permutation null
was not recomputed after singleton exclusion and is omitted.}
\label{tab:results}
\setlength{\tabcolsep}{3.5pt}
\begin{tabular}{@{}lllcccccl@{}}
\toprule
\textbf{Analysis} & \textbf{Source} & \textbf{Sec.} & $K$ & $D_{\mathrm{obs}}$ & $z_{\mathrm{perm}}$ & $z_{\mathrm{str}}$ & $\hat{p}_{\mathrm{str}}$ & \textbf{Outcome} \\
\midrule
NKI 1D   & \cite{Lum2013}   & \ref{sec:res-nki}  &  6 & 0.159 & 4.40  &  0.13    & 0.24 & Not distinguished \\
NKI 2D   & \cite{Lum2013}   & \ref{sec:res-nki}  &  6 & 0.159 & 5.91  &  0.91    & 0.10 & Not distinguished \\
Voting   & \cite{Lum2013}   & \ref{sec:res-voting} &  3 & 0.630 & 22.36 & $-1.72$  & 0.95 & Not distinguished \\
NBA (all) & \cite{Lum2013}  & \ref{sec:res-nba}  & 40 & 5.343 & 3.30  &  5.51    & $<0.005$ & See text \\
NBA (multi) & \cite{Lum2013} & \ref{sec:res-nba}  & 14 & 1.803 & 3.29  & $-2.31$ & 1.00 & Not distinguished \\
LGG (all) & \cite{RabadanCamara2020} & \ref{sec:res-lgg} & 13 & 1.364 & 6.08  &  5.93    & $<0.01$ & See text \\
LGG (multi) & \cite{RabadanCamara2020} & \ref{sec:res-lgg} & 12 & 0.595 & ---   &  0.81    & 0.20 & Not distinguished \\
\bottomrule
\end{tabular}
\end{table}

\subsection{Breast cancer gene expression}\label{sec:res-nki}

The breast cancer analysis of~\cite{Lum2013} applies Mapper to the
Netherlands Cancer Institute (NKI) gene expression
dataset~\citep{vanDeVijver2002}.  The resulting community structure
is interpreted as evidence of clinically meaningful subtypes, including a
subgroup with high survival rates that does not correspond to any
standard clinical classification.  We ask whether this community
differentiation exceeds what the covariance structure alone can produce.

We replicate this analysis ($n = 336$ samples, $p = 1{,}500$
highest-variance genes, Pearson correlation distance) in two
configurations.  The first is a 1D Mapper with $L^\infty$ centrality filter
(resolution~30, gain~3.0, 5 histogram bins; 6 communities, zero
singletons).  The second adds a clinical survival indicator as
a second filter (resolution~$30 \times 5$, gain~3.0; 6 communities,
zero singletons).  For each, $B = 200$ null replicates are drawn from
$\mathcal{N}_{1500}(0, \hat{\Sigma})$, and we also compute 1,000
label-permutation replicates.  Full data-preprocessing, parameter, and
pipeline details are in \cref{sm:replication-details}.

In the one-dimensional analysis,
the Louvain algorithm detects 6 communities (sizes 109, 83, 50, 43, 27,
and 24) with an observed dissociation of $D_{\mathrm{obs}} = 0.159$.
The label-permutation null yields $z_{\mathrm{perm}} = 4.40$
($p < 0.001$), indicating that the community assignments are far more
differentiated than random label assignments would produce. Under the
label-permutation null, the specific mapping of data points to
communities carries significant information.

Under the covariance-preserving null, however, the picture changes.
The 200 null replicates drawn from $\mathcal{N}_{1500}(0,
\hat{\Sigma})$ produce a null distribution of dissociation values with
mean $\bar{D}^* = 0.136$ and standard deviation $s_{D^*} = 0.173$.
The observed dissociation of 0.159 falls well within this null
distribution, yielding $z_{\mathrm{str}} = 0.13$
($\hat{p} = 0.24$).  Because the null dissociation distribution is
right-skewed, the $z$-score should be interpreted cautiously.  The Monte Carlo $p$-value
$\hat{p} = 0.24$, which makes no distributional assumption, is the
primary inferential quantity.  The Mapper community structure observed
in the NKI data is not distinguishable from what arises from a single
multivariate normal distribution with the same covariance matrix.

The two-dimensional analysis tells a similar story: 6 communities,
$z_{\mathrm{perm}} = 5.91$ ($p < 0.001$), but
$z_{\mathrm{str}} = 0.91$ ($\hat{p} = 0.10$).  Because the null
features are drawn independently of the fixed survival indicator, any
correlation between gene expression and survival present in the real
data is severed in the null; this could inflate the null $z$-score, so
the non-rejection is conservative in this respect.  Two supplementary diagnostics agree qualitatively: replacing the
dissociation metric with the maximum over individual feature differences
($D_{\max}$), and comparing the observed graph modularity against its
null distribution, both support non-rejection (\cref{sm:empirical}).  The Mapper community structure is not
distinguishable from the covariance-preserving Gaussian benchmark.  This
does not imply that breast cancer subtypes do not exist, but that Mapper
has not provided evidence for their existence beyond what the
second-order geometry explains.  This non-rejection is stable across
sensitivity checks: none of
51 random feature splits rejects, and 11 of 12 Mapper parameter
configurations fall below the critical value (\cref{sm:robustness}).

\subsection{U.S.\ Congressional voting}\label{sec:res-voting}

The voting analysis of~\cite{Lum2013} applies Mapper to U.S.\ House
roll-call voting records and finds communities that align with political
parties, including a small group of moderate ``bridge'' members
connecting the two partisan clusters.  We ask whether the observed
community differentiation goes beyond what the correlation structure of
the votes already explains.

We replicate this analysis for the 110th Congress (2008).  The data
consist of $n = 438$ House members and $p = 688$ rollcall votes coded as
$+1$ (yea), $-1$ (nay), or $0$ (absent/abstain), with Pearson
correlation distance and PCoA filters (resolution~20, gain~4.5,
5 histogram bins).  The Mapper graph yields 3 Louvain communities with
no singleton communities.

Because $p > n$, the sample covariance matrix is singular.  We generate
null samples via the eigendecomposition of~$\hat{\Sigma}$, drawing only
in the 437-dimensional subspace where the data vary.  This avoids the
artificial isotropic component that ridge regularization would introduce
for a matrix with 251 zero eigenvalues.  We generate $B = 200$ structured
null replicates and 1,000 label-permutation replicates.  Full details
are in \cref{sm:replication-details}.

The three communities align closely with Democrats, Republicans, and a
small moderate bridge group, with an observed dissociation of
$D_{\mathrm{obs}} = 0.630$.  The label-permutation null yields an
extremely high $z_{\mathrm{perm}} = 22.36$ ($p < 0.001$), reflecting
the strong partisan structure in the data.

Under the covariance-preserving null, the observed dissociation falls
\emph{below} the null mean, yielding $z_{\mathrm{str}} = -1.72$
($\hat{p} = 0.95$).  Synthetic data drawn from a
single multivariate normal with the same covariance matrix produce
\emph{more} community differentiation in Mapper than the real voting
data.  Why?  The real voting data are strongly bimodal (two parties),
so Mapper produces a small number of internally homogeneous communities.
The Gaussian null, which preserves the elongated covariance structure but
has no bimodal gap, distributes points more smoothly along the dominant
axis.  Mapper then resolves a gradient of communities spanning the full
range, and the most extreme pair can differ more than the two-party pair
in the real data.

This result illustrates an important distinction. The partisan
communities are visually compelling and substantively meaningful---they
really do correspond to parties---but the observed community
differentiation does not exceed what the covariance-preserving null
can generate. Therefore, the Mapper graph does not add inferential
evidence for latent political subtype structure beyond the benchmark
provided by the covariance-matched Gaussian null.

\subsection{NBA player performance}\label{sec:res-nba}

The NBA analysis of~\cite{Lum2013} applies Mapper to per-minute
statistics for NBA players and identifies communities interpreted as
distinct player types, including groups of shooters, rebounders, and
versatile ``all-around'' players.  We ask whether this community
structure reflects genuine player-type differentiation beyond what the
joint distribution of statistics can produce.

We replicate this analysis: $n = 452$ players from
the 2010--2011 season, $p = 7$ per-minute statistics, Euclidean distance,
PCoA filters.  We analyze under two configurations:
Configuration~1 (resolution~30, gain~2.0, 15 histogram bins;
40 communities, 26 singletons) and Configuration~2 (resolution~20,
gain~2.5, 3 bins; 7 communities, 1 singleton).  The $7 \times 7$ covariance is well-conditioned.  We
report results both including and excluding singleton communities,
because the dissociation metric's maximum-over-pairs construction can be
dominated by singletons with extreme statistics.  For each
configuration, $B = 200$ null replicates and 1,000 label-permutation
replicates.  Full details are in \cref{sm:replication-details}.

The NBA analysis reveals an important methodological finding about the
dissociation metric. Under Configuration~1,
the Mapper graph produces 40 Louvain communities, of which 14 contain
more than one player (sizes 106, 86, 43, 41, 39, 32, 31, 12, 10, 9,
8, 4, 3, and 2; the remaining 26 are singletons; modularity 0.707).

When all 40 communities are included in the dissociation computation,
the observed dissociation is $D_{\mathrm{obs}} = 5.343$, and the
structured null yields $z_{\mathrm{str}} = 5.51$ ($\hat{p} < 0.005$).
However, when singleton communities are excluded, the observed
dissociation drops to $D_{\mathrm{obs}} = 1.803$, and the structured
null yields $z_{\mathrm{str}} = -2.31$ ($\hat{p} = 1.00$).  The
qualitative conclusion reverses: the community differentiation among
multi-player communities is \emph{below} the null mean.

Configuration~2 (resolution~20, gain~2.5; 7 communities, 1 singleton)
tells the same story: $z_{\mathrm{str}} = 8.01$ with all communities,
$z_{\mathrm{str}} = -0.85$ ($\hat{p} = 0.80$) without the singleton.

The mechanism is straightforward.  Singleton communities represent
individual players with potentially extreme per-minute statistics.
Because the dissociation metric takes the maximum over all community
pairs, a single pair involving a singleton can dominate the test
statistic.  In the null replicates, such extreme singletons are less
likely, so the null has lower dissociation.  The rejection is driven by
sensitivity to outlier communities, not by systematic differentiation.
This pattern holds across all 20 Mapper parameter configurations and
all 51 random feature splits: every configuration rejects with
singletons and none rejects without them (\cref{sm:robustness}).
The $D_{\max}$ and null modularity diagnostics described above also
agree qualitatively (\cref{sm:empirical}).

\subsection{Lower-grade glioma gene expression}\label{sec:res-lgg}

The lower-grade glioma (LGG) analysis of~\cite{RabadanCamara2020}
applies Mapper to gene expression data from The Cancer Genome Atlas and
identifies community structure related to well-established molecular
subtypes, including mutations in the IDH gene and co-deletion of
chromosome arms 1p and 19q.  Unlike the preceding three analyses, this
case comes from a separate study and uses a different Mapper pipeline,
so it provides an additional check on the framework's generalizability
and scope.  It is also a useful stress test, because many LGG subtype
signals contribute to the pooled covariance and may therefore be
absorbed by the covariance-preserving null rather than detected by it.

The data consist of $n = 534$ tumor samples and $p = 4{,}500$
highest-variance genes, with Pearson correlation distance.  The filter
functions differ from the~\cite{Lum2013} analyses: we build a
$k$-nearest-neighbor graph ($k = 30$), compute shortest-path distances,
and apply multidimensional scaling (MDS) to obtain two filter
coordinates.  We run the TDAmapper software package with fixed-width
binning (resolution~30, 75\% overlap, 10 histogram bins).

Because $p \gg n$, the sample covariance is singular.  We therefore
apply the ridge regularization described in \cref{supp:ridge} to obtain
a positive-definite covariance for null sampling.  We generate $B = 100$
null replicates, reduced from 200 due to computational cost, and 200
label-permutation replicates.  Full details are in
\cref{sm:replication-details}.

The LGG analysis reveals the same singleton sensitivity observed in
the NBA data.  The Mapper graph produces 13 Louvain communities (sizes
95, 89, 73, 60, 60, 48, 32, 27, 19, 10, 10, 9, and 1; these sum to
533 because one sample does not appear in any Mapper vertex and is
excluded from the dissociation computation) with modularity 0.738
and an observed dissociation of $D_{\mathrm{obs}} = 1.364$.  One
community is a singleton.  The label-permutation null yields
$z_{\mathrm{perm}} = 6.08$ ($p < 0.01$).

Under the structured null, the 100 replicates from $\mathcal{N}_{4500}(0,
\hat{\Sigma} + \epsilon I)$ produce a null distribution with mean
$\bar{D}^* = 0.565$ and standard deviation $s_{D^*} = 0.135$, yielding
$z_{\mathrm{str}} = 5.93$ ($\hat{p} < 0.01$; no null replicate
exceeded the observed statistic).  However, when the singleton community
is excluded, the observed dissociation drops from 1.364 to 0.595, and
the structured null yields $z_{\mathrm{str}} = \singletonExcludedLGGz$
($\hat{p} = 0.20$).  Thus, as in the NBA analysis, the apparent
rejection is driven by a singleton community rather than by systematic
differentiation among the multi-sample communities.  The
singleton-excluded $\hat{p} = 0.20$ carries additional uncertainty from the
smaller number of null replicates ($B = 100$ vs.\ 200 for the other
analyses).

Three points clarify how to interpret this non-rejection.  First, ridge regularization adds small but nonzero variance in the
$4{,}500 - 533 = 3{,}967$ null directions where the sample covariance
is zero, making the null data more isotropic and diluting the dominant
covariance directions that drive community differentiation.  The
resulting null Mapper graphs tend to have lower dissociation, making
the null \emph{easier} to beat, so this bias works against the
non-rejection finding.

Second, gene expression data can exhibit heavy tails and skewness even
after log-transformation, and our simulation study
(Section~\ref{sec:sim-typeI}) shows that such departures from normality
can produce substantial overrejection (79.5\% under a multivariate~$t$
with $\mathrm{df} = 5$).  This matters because the singleton-excluded
LGG $z$-score of $\singletonExcludedLGGz$ is positive, whereas the simulated mean-shift
mixtures produce near-zero or negative $z$-scores.  A plausible
explanation is non-Gaussianity rather than residual subtype signal:
non-Gaussianity inflates $z$-scores under the Gaussian null.

Third, lower-grade gliomas have well-established molecular
subtypes~\citep{Cancer2015lgg} with both distinct mean expression
profiles and distinct covariance structures.  Our simulation study
(Section~\ref{sec:sim-power}) shows, however, that the dissociation
metric cannot detect mean-shift mixtures whose between-component
structure is absorbed by the sample covariance, nor can it detect
mixtures differing only in covariance structure.  For real biological
subtypes like those in LGG, between-subtype mean differences contribute
to the pooled covariance estimate.  The covariance-preserving null can
therefore reproduce much of the variation associated with those
subtypes.  The singleton-excluded non-rejection is consistent with the
scope limitations seen in the simulations.

Despite the methodological differences from the preceding analyses, the
structured null test produces an interpretable result and reveals the
same singleton sensitivity as the NBA analysis, suggesting that both the
framework and the singleton finding generalize beyond any single Mapper
pipeline.


\section{Discussion and conclusions}\label{sec:discussion}

Across four empirical case studies spanning cancer genomics, political
science, and sports analytics, no analysis produces community
differentiation that exceeds the structured null at $\alpha = 0.05$
once singleton communities are accounted for.
The breast cancer and Congressional voting analyses are clearly not
distinguished from the null.  The NBA and LGG analyses initially appear
to reject, but both rejections are driven by singleton communities;
when these are excluded, the remaining community structure falls within
the null distribution.  Two results merit closer attention
(NKI 2D: $\hat{p} = 0.10$; LGG singleton-excluded:
$\hat{p} = 0.20$).  For both NBA and NKI, the non-rejection findings
are stable across alternative feature splits and Mapper parameter variation
(\cref{sm:robustness}).  In every analysis, the label-permutation null
is strongly significant ($z_{\mathrm{perm}}$ ranging from 3.29 to
22.36) while the structured null is not, illustrating how much community
differentiation covariance geometry alone can explain.  Under the independence-based null of~\cite{Lum2013}, all three of
their analyses are significant; under the covariance-preserving null,
none exceeds the benchmark once singletons are excluded.

Non-rejection does not mean the identified communities are
uninteresting.  The voting communities, for example, correspond to real
political parties.  Non-rejection means that the community
differentiation does not require subgroup structure beyond what the
covariance geometry already provides.  The singleton exclusion is a
diagnostic decision motivated by the observed discrepancy, not
pre-specified, and should be interpreted as a sensitivity analysis
rather than a confirmatory result.  More broadly, the framework not only
provides a hypothesis test but also reveals where apparent signal
originates, enabling practitioners to distinguish systematic community
differentiation from artifacts of extreme observations.

The framework has several limitations.  The Gaussian null preserves
only covariance structure, so it can overreject under non-Gaussian
unimodal data: 79.5\% under a multivariate $t$ with $\mathrm{df} = 5$,
23.0\% under skewed data (Section~\ref{sec:sim-typeI}).  This is
particularly relevant to gene expression datasets with heavy-tailed or
skewed marginals.  The null also treats $\hat{\Sigma}$ as known,
ignoring estimation uncertainty; ridge regularization addresses
singularity but adds isotropic variance, making the null easier to beat.
For non-rejection findings this bias is conservative, but simulations
confirm well-controlled Type~I error only up to $p/n \approx 1.67$,
while the empirical genomic analyses have higher ratios (NKI: 4.5; LGG:
8.4).  Shrinkage estimators~\citep{LedoitWolf2004} or factor models
could improve the null in high-dimensional settings.  The test also
conditions on fixed Mapper parameters and does not control the
family-wise error rate across the parameter space; adaptive
covers~\citep{Alvarado2025, Tao2025} could address this.  Our
replication of~\cite{Lum2013} is approximate: Ayasdi's internal
parameterization differs from our open-source implementation, and the
NKI sample size differs (336 vs.\ 272).  The consistency of
non-rejection across configurations, feature splits, and the independent
LGG analysis suggests that the conclusions are not sensitive to these choices
(\cref{sm:replication-scope}).

The dissociation metric itself is a further limitation.  It is not
sensitive to all forms of topological structure; alternative statistics
based on Betti numbers, modularity, or persistence could target
different aspects of Mapper output.  A robustness check using
$D_{\max}$, which maximizes over individual feature differences rather
than block averages, agrees qualitatively in every analysis
(\cref{sm:empirical}).  We use a one-sided test ($z > 1.645$) because
the scientific question is whether the data contain more structure than
the null can explain.  The voting result ($z = -1.72$) illustrates that
bimodal data can produce \emph{less} differentiation than the smooth
Gaussian gradient.  More fundamentally, the dissociation metric misses
mean-shift mixtures absorbed by the covariance, is blind to
covariance-only alternatives (Section~\ref{sec:sim-power}), and is
highly sensitive to singleton communities, as the NBA and LGG analyses
show.  Developing statistics that are robust to small communities and
responsive to a broader range of alternatives would extend the
framework's practical reach.

For practitioners, three recommendations follow.  Mapper-based
subtype claims should be accompanied by null-model tests that preserve
the covariance structure of the data, since label-permutation baselines
are uninformative (Theorem~\ref{thm:permutation}).  Results should
be reported both with and without singleton communities, since a large
discrepancy signals statistic instability.  Rejections in
datasets with non-Gaussian marginals should be interpreted cautiously.
Computational cost scales with pipeline complexity; for the LGG
analysis, 100 null replicates are already expensive.

Mapper remains a valuable exploratory tool whose ability to produce
interpretable graph summaries of high-dimensional data is not diminished
by our findings.  What our results show is that Mapper's exploratory
output should not be interpreted as confirmatory evidence for the
existence of subtypes without additional statistical testing.  The
general principle---testing observed combinatorial summaries against a
covariance-matched null---could inform null-model development for other
TDA methods such as Reeb graphs or simplicial complexes, though the
appropriate test statistics and null distributions would need to be
developed for each setting.

\section*{Acknowledgments}
The author used Claude (Anthropic, Claude Opus 4) for manuscript editing
assistance, including sentence-level revisions, style auditing, and
checking for internal consistency.  All output was verified by the
author, who takes full responsibility for the content.
All datasets used in this study are publicly available: the NKI breast
cancer dataset via the \texttt{breastCancerNKI} Bioconductor package,
Congressional voting data from VoteView (\texttt{voteview.com}), NBA
player statistics from Basketball Reference
(\texttt{basketball-reference.com}), and TCGA lower-grade glioma data
from the CamaraLab GitHub repository.  All code for replication is
available at \texttt{[URL to be added upon publication]}.

\appendix

\section{Null-model audit of published Mapper applications}\label{sm:audit}

Table~\ref{tab:null-audit-sm} summarizes a targeted but not exhaustive
audit of null-model testing across published Mapper applications that
claim to discover subtypes or structural groups.  For each analysis, we record whether the Mapper graph
structure was tested against a null model and, if so, what null was
used.  ``Independence-based null'' refers to a null that tests only
against unstructured noise, destroying the covariance geometry of the
data.  Phase-randomization nulls preserve aspects of temporal dependence,
such as the power spectrum and autocorrelation structure of each time
series, making them analogous in spirit to the covariance-preserving
logic proposed in the main text for cross-sectional data.

\begin{table}[ht]
\centering
\scriptsize
\caption{Null-model testing in published Mapper applications relevant
to subtype discovery, structural group claims, or Mapper validation.}
\label{tab:null-audit-sm}
\begin{tabular}{@{}lllll@{}}
\toprule
\textbf{Reference} & \textbf{Domain} & \textbf{Structure claimed} & \textbf{Null?} & \textbf{Null model \& assessment} \\
\midrule

\cite{Nicolau2011}
  & Breast cancer
  & c-MYB+ subgroup
  & No
  & Biological validation only \\[4pt]

\cite{Lum2013}
  & Breast cancer
  & Y-shape, flares
  & Yes
  & Indep.\ Gaussians, common $\sigma^2$ \\[4pt]

\cite{Lum2013}
  & Pol.\ science
  & Partisan subcommunities
  & Yes
  & Same independence-based null \\[4pt]

\cite{Lum2013}
  & Sports
  & 13 player types
  & Yes
  & Same independence-based null \\[4pt]

\cite{Li2015}
  & Type 2 diabetes
  & 3 diabetes subgroups
  & No
  & External genetic validation \\[4pt]

\cite{Yao2009}
  & Protein folding
  & Two refolding pathways
  & No
  & Kinetic (Markov) validation \\[4pt]

\cite{Rizvi2017}
  & scRNA-seq
  & Transient cell states
  & Partial
  & PH test for loops only \\[4pt]

\cite{Saggar2018}
  & Neuroimaging
  & Task communities
  & Partial
  & Phase-randomization null \\[4pt]

\cite{Geniesse2022}
  & Neuroimaging
  & Replicates Saggar
  & Partial
  & Phase-randomization null \\[4pt]

\cite{Hasegan2024}
  & Neuroimaging
  & Parameter sensitivity
  & No
  & GOF criteria, no null model \\[4pt]

\cite{Rafique2020}
  & Cancer (TCGA)
  & Cancer subtypes
  & No
  & Survival analysis only \\[4pt]

\cite{Lum2012}
  & Pol.\ science
  & Voting fragmentation
  & No
  & Descriptive analysis only \\[4pt]

\cite{Wamil2023}
  & Type 2 diabetes
  & 4 diabetes phenotypes
  & No
  & Survival analysis only \\

\bottomrule
\end{tabular}

\medskip
\raggedright
\scriptsize
\textbf{Abbreviations and terms.}
Pol.: political.
c-MYB+: subgroup defined by expression of the \emph{c-MYB} oncogene.
scRNA-seq: single-cell RNA sequencing.
PH: persistent homology.
GOF: goodness of fit.
TCGA: The Cancer Genome Atlas.
Phase-randomization null: a null that randomly permutes the Fourier
phases of each time series, preserving its power spectrum and
autocorrelation structure.
\end{table}

\section{Replication methodology}\label{sm:replication}

Our replications target four published Mapper analyses: three from
\cite{Lum2013} (breast cancer, congressional voting, and NBA player
performance) and one from~\cite{RabadanCamara2020} (lower-grade glioma).
Because the original~\cite{Lum2013} analyses used a proprietary platform
whose parameterization differs from open-source tools, we first describe
how we translated the Mapper pipeline, then give the full parameter
specifications for each case study.

\subsection{Ayasdi-to-open-source Mapper translation}\label{sm:ayasdi}

The original analyses in~\cite{Lum2013} were conducted in Ayasdi, a
proprietary platform whose internal parameterization differs from
open-source Mapper implementations.  Ayasdi's ``Resolution'' and ``Gain''
parameters do not correspond directly to the number of intervals and
overlap percentage in open-source tools such as TDAmapper
\citep{Tauzin2021} or KeplerMapper \citep{vanVeen2019}.  Moreover,
Ayasdi's ``Equalized'' mode uses quantile-based binning, in which each
cover interval is chosen to contain approximately the same number of data
points.  Standard open-source implementations instead use equal-width
intervals, which can produce nearly empty bins in the tails of the filter
distribution and excessive fragmentation of the graph.

To address this, we implement a custom equalized (quantile-based) Mapper
that mimics Ayasdi's binning strategy.  Cover intervals are defined by
equally spaced quantiles of the filter function, so that each interval
contains roughly the same number of points regardless of the filter's
marginal distribution.  Within each interval (or, in the two-dimensional
case, each grid cell), we apply single-linkage hierarchical clustering
and cut the dendrogram at the first gap (first empty bin) in a histogram
of merge heights, following the original algorithm of~\cite{Singh2007}.
The number of histogram bins is a hyperparameter reported for each
analysis in the main text.  This matches the clustering method reported in
\cite{Lum2013}.

Because the parameter mapping between Ayasdi and open-source Mapper is
not exact, we cannot simply use the parameter values reported in
\cite{Lum2013} (e.g., Resolution~70 for breast cancer, Resolution~120
for voting).  Applying these values to our equalized Mapper produces
graphs dominated by singleton vertices---far more fragmented than the
graphs shown in~\cite{Lum2013}.  To find parameter settings that produce
comparable graph structures, we conducted systematic parameter sweeps for
each dataset, varying resolution, gain, and the number of histogram bins
used in the clustering step.  We selected parameters that produce graphs
with community structure consistent with the results reported in
\cite{Lum2013}: similar numbers of communities, similar community size
distributions, well-defined community structure, and minimal singleton
vertices.
The specific parameters for each analysis are reported in the main text.

This parameter-selection step is an unavoidable consequence of
replicating proprietary-software analyses with open-source tools.  All null-model comparisons use identical Mapper parameters for the
observed data and each null replicate, so the Monte Carlo comparison
is valid conditional on the selected parameters.  We interpret the
resulting tests as conditional on this replication-oriented
parameter-selection step.

\subsection{Detailed replication methods}\label{sm:replication-details}

The parameter specifications and implementation details below
supplement the methodological summaries in the main text, providing
the complete recipes needed for reproduction.

For each structured-null replicate, we recompute the distance matrix,
filter coordinates, Mapper graph, Louvain communities, point-to-community
assignments, and test statistic from the synthetic data, using the same
pipeline and parameters as in the observed analysis.  Exceptions (such as
the NKI survival indicator being reused with fresh jitter) are noted
in the relevant case study.  Graph communities are detected using
\texttt{igraph::cluster\_louvain}, which is stochastic (Louvain
optimization involves random tie-breaking).  The same algorithm is
applied to both observed and null Mapper graphs.  For each null
replicate, the random seed is set before the full pipeline call, so
Louvain tie-breaking is determined by the replicate seed; the observed
Mapper graph uses a single Louvain realization.

\subsubsection*{Breast cancer gene expression (NKI)}

We obtained the NKI dataset from the \texttt{breastCancerNKI}
Bioconductor package.  After removing samples with more than 10\%
missing values across genes, our dataset contains $n = 336$ samples.
\cite{Lum2013} report $n = 272$.  The original NKI cohort
\citep{vanDeVijver2002} contained 295 samples; the
\texttt{breastCancerNKI} package provides an expanded aggregate dataset,
which explains why our sample size exceeds both the original cohort and
the subset used in~\cite{Lum2013}.  We selected the
$p = 1{,}500$ genes with the highest sample variance (computed with
\texttt{na.rm = TRUE} to handle missing entries), matching the description
in~\cite{Lum2013} of using the ``top most varying genes.''
Their Supplementary Figure~S1 shows $\approx$1,553 genes.  Remaining missing
values were imputed with column means, and the resulting
$336 \times 1{,}500$ matrix was centered and scaled to unit variance.

Following~\cite{Lum2013}, we used Pearson correlation distance:
$d(x_i, x_j) = 1 - \mathrm{cor}(x_i, x_j)$, computed across the 1,500
gene expression values for each pair of samples.

For the one-dimensional analysis, the filter function is the $L^\infty$
centrality of each sample: $f(x_i) = \max_j\, d(x_i, x_j)$.  For the
two-dimensional analysis, the first filter is again $L^\infty$
centrality; the second filter is the clinical survival indicator
(\texttt{event\_death}: 1 if the patient died during follow-up, 0
otherwise), as reported in Figure~2 of~\cite{Lum2013}.  Because this
variable is binary, we add independent Gaussian jitter with standard
deviation 0.01 to avoid degenerate quantile bins.

For the one-dimensional analysis, we use resolution $N = 30$, gain
$g = 3.0$, and 5 histogram bins for the clustering step.  This produces
a comparable graph (6 communities with sizes 109, 83, 50, 43, 27,
and 24; modularity 0.648; zero singleton vertices).  For the
two-dimensional analysis, the first filter uses resolution~30 and the
second filter uses resolution~5; gain is 3.0 for both filters, and
clustering uses 5 histogram bins.  This produces 6 communities (sizes 86,
71, 54, 50, 43, and 32; zero singletons).

Because $p = 1{,}500 > n = 336$, the sample covariance matrix
$\hat{\Sigma}$ is singular, with rank at most $n - 1 = 335$.  We
generate null replicates via the eigendecomposition of
$\hat{\Sigma}$: decompose $\hat{\Sigma} = V \Lambda V^\top$ with
$r$ positive eigenvalues $\lambda_1, \ldots, \lambda_r$, draw
$Z \in \R^{n \times r}$ with iid standard normal entries, and set
$X^* = Z\, \Lambda_r^{1/2}\, V_r^\top$.  This draws from the
degenerate Gaussian supported on the column space of $\hat{\Sigma}$
(in our implementation, \texttt{MASS::mvrnorm} performs this
eigendecomposition internally).  For both analyses, we generate
$B = 200$ structured null replicates.  In the two-dimensional
analysis, the survival indicator is reused with fresh jitter for each
null replicate.  We also compute 1,000 label-permutation replicates.

\subsubsection*{U.S.\ Congressional voting}

Rollcall voting data were obtained from VoteView
(\texttt{voteview.com}).  For the 110th Congress (2008), the dataset
contains $n = 438$ House members and $p = 688$ rollcall votes after
filtering to members who participated in at least 10\% of votes.  Each
vote is coded as $+1$ (yea), $-1$ (nay), or $0$ (absent or abstain).
We centered and scaled all votes to unit variance.
Because $p > n$, the sample covariance matrix is singular, with rank at
most $n - 1 = 437$.

We used Pearson correlation distance between members.  The two filter
functions are the first and second coordinates of a principal coordinate
analysis (PCoA) applied to the correlation distance matrix.

We use resolution~20 for both filters, gain~4.5, and 5 histogram bins.
This produces a clean partisan structure (3 communities corresponding to
Democrats, Republicans, and a moderate bridge group; zero singletons).

Both the observed Mapper graph and every null replicate use the full set
of $p = 688$ scaled votes.  We generate null samples via the
eigendecomposition of~$\hat{\Sigma}$, drawing random variation only in
the 437-dimensional subspace where the data actually vary, exactly as
described for NKI above.  We generate $B = 200$ structured null
replicates and 1,000 label-permutation replicates.

\subsubsection*{NBA player performance}

Player statistics for the 2010--2011 NBA season were obtained from
Basketball Reference.  Following~\cite{Lum2013}, we compute per-minute
rates for $p = 7$ statistics: total rebounds, assists, turnovers, steals,
blocked shots, personal fouls, and points scored.  Players with fewer
than 1 total minute played are excluded, yielding $n = 452$ players.

We use variance-normalized Euclidean distance and PCoA filters.
Configuration~1 (paper-matching) uses resolution~30, gain~2.0, and 15
histogram bins; Configuration~2 (sweep-recommended) uses resolution~20,
gain~2.5, and 3 histogram bins.

The $7 \times 7$ sample covariance matrix is well-conditioned.  For each
configuration, we compute $B = 200$ structured null replicates and 1,000
label-permutation replicates.

\subsubsection*{Lower-grade glioma (LGG)}

RNA-seq gene expression data for TCGA LGG samples were obtained from
the CamaraLab GitHub repository.  The repository provides
log$_2$-transformed expression values, which we used without additional
transformation or centering/scaling in order to match the preprocessing
used in the original pipeline.  We applied no additional sample filtering
beyond what the repository provides.  We selected the
$p = 4{,}500$ genes with the highest variance, yielding $n = 534$
samples.  We use Pearson correlation distance.  The filter functions are
the first two coordinates of classical MDS applied to shortest-path
distances in a directed $k$-nearest-neighbor graph ($k = 30$), built
using \texttt{cccd::nng} following the original code
of~\cite{RabadanCamara2020}.  Any infinite shortest-path distances
(from unreachable pairs in the directed graph) are replaced by twice the
maximum finite distance.  Following the original code
of~\cite{RabadanCamara2020}, we pass the directed shortest-path matrix
directly to \texttt{cmdscale}.  Classical MDS is defined for symmetric
dissimilarity matrices; the directed shortest-path matrix is in general
asymmetric.  In practice, \texttt{cmdscale} operates on the double-centered matrix
$-\tfrac{1}{2} H D H$ where $H = I - \mathbf{1}\mathbf{1}^\top/n$ and
$D$ is the input; when $D$ is asymmetric, this implicitly produces a
symmetric kernel.  We retain this step for faithful replication of the
published pipeline rather than as a canonical MDS construction.

We use TDAmapper with fixed-width binning, resolution~30 for both
filters, 75\% overlap, and 10 histogram bins.  The sample covariance
matrix is singular ($p = 4{,}500 > n = 534$); we apply ridge
regularization with $\epsilon$ chosen as in \cref{supp:ridge}.  We
generate $B = 100$ structured null replicates and 200 label-permutation
replicates.

\section{Simulation study details}\label{sm:simulation}

The simulation study evaluates the structured null model's Type~I error
calibration under the null hypothesis and its behavior under
mean-shift and covariance-shift alternatives.  All scenarios share a
common Mapper pipeline; they differ only in the data-generating process.

\subsection{Data-generating process specifications}\label{sm:dgp}

Each simulation scenario generates $R = 200$ independent datasets of
$n = 300$ observations.  For the skewed and mixture DGPs, we center and scale each simulated
dataset featurewise before computing distances,
filters, Mapper graphs, and the covariance matrix used for null sampling.
For the mixture DGPs, $\delta$ denotes the component separation before
this centering and scaling step.
The spherical Gaussian, correlated Gaussian, and multivariate~$t$ DGPs
are already zero-mean and are not rescaled.  The Mapper pipeline uses two-dimensional
equalized Mapper with PCoA filters, resolution~15 per filter,
gain~2.0, and 10 histogram bins for clustering.  Each test uses
$B = 50$ null replicates.  As a Monte Carlo precision check, we also
run the correlated Gaussian $p = 10$ scenario with $B = 100$; the
rejection rate is essentially unchanged (5.0\% vs.\ 6.0\%), confirming
that $B = 50$ provides adequate calibration.

The data-generating processes are:

\begin{enumerate}
\item \textbf{Spherical Gaussian} ($p = 10$).  $X \sim \mathcal{N}_{10}(0, I_{10})$.

\item \textbf{Correlated Gaussian} ($p \in \{10, 50, 500\}$).
$X \sim \mathcal{N}_p(0, \Sigma)$ with block covariance: features are
split into two equal-sized groups, features within the same group have
pairwise correlation $\rho = 0.5$, and features in different groups are
independent.

\item \textbf{Multivariate $t$} ($p = 10$, $\mathrm{df} = 5$).
$Z \sim \mathcal{N}_p(0, \Sigma)$, $V \sim \chi^2_\nu$ independent of~$Z$,
$X = Z / \sqrt{V/\nu}$ with $\nu = 5$.  The marginal covariance is
$\frac{\nu}{\nu - 2}\Sigma = \tfrac{5}{3}\Sigma$; the structured null
estimates this inflated covariance from the sample covariance, so the
null matches the second-order scale of the data but not its heavy-tailed
distribution.

\item \textbf{Skewed unimodal} ($p = 10$).  $Z \sim \mathcal{N}_p(0,
\Sigma)$ followed by componentwise $\exp(0.5\, Z_j)$, then re-centered
and scaled.

\item \textbf{All-feature mean-shift mixture} ($p \in \{10, 50\}$).
Equal-weight mixture: $\mathcal{N}_p(+\delta \mathbf{1}/2, \Sigma)$ and
$\mathcal{N}_p(-\delta \mathbf{1}/2, \Sigma)$ with
$\delta \in \{0.5, 1.0, 2.0\}$, where $\Sigma$ is the block covariance
from item~2 with $\rho = 0.5$.

\item \textbf{Sparse mean-shift mixture} ($p = 50$).  Equal-weight
mixture with $\Sigma$ as in item~2 ($\rho = 0.5$), but only the first
$k = 5$ of 50 features are shifted by $\pm \delta / 2$
($\delta \in \{1.0, 2.0\}$).

\item \textbf{Heterogeneous-covariance mixture} ($p \in \{10, 50\}$).
Equal-weight, zero-mean mixture with different within-component
correlation: component~1 has correlated features ($\rho$) in the first
half and independent in the second; component~2 has the reverse.
$\rho \in \{0.5, 0.8, 0.9\}$.
\end{enumerate}

\section{Robustness and parameter sensitivity}\label{sm:robustness}

The main-text analyses each use a single set of Mapper parameters and a
single deterministic feature split.  Here we vary both to assess
sensitivity: first across random feature splits, then across a grid of
Mapper resolution and gain values.

\subsection{Feature-split robustness}\label{sm:feature-splits}

To assess whether results depend on the particular choice of feature
split, we repeat the structured null test using 50 independent random
50/50 feature splits in addition to
the original odd/even split, for both the NBA and NKI breast cancer
datasets.

For the NBA dataset under Configuration~1 (resolution~30, gain~2.0),
when all 40 communities (including 26 singletons) are included,
$z$-scores range from 0.48 to 9.60 with mean 5.60, and 50 of 51 splits
reject at $\alpha = 0.05$.  When singleton communities are excluded, all
51 splits produce negative $z$-scores (range $-3.65$ to $-1.07$, mean
$-2.48$), and none rejects.  Configuration~2 tells the same story: 51 of
51 splits reject with all communities (mean $z = 6.22$); only 1 of 51
rejects without singletons (mean $z = -0.97$).  The singleton artifact
is robust across feature splits.

For the NKI breast cancer dataset (1D Mapper, resolution~30, gain~3.0),
the 51 $z$-scores range from $-0.07$ to 0.77 with mean 0.17.  None of
the 51 splits rejects at $\alpha = 0.05$.

\subsection{Parameter sensitivity}\label{sm:param-sensitivity}

For the NBA dataset, we test resolution $\in \{10, 15, 20, 25, 30\}$ and
gain $\in \{1.5, 2.0, 2.5, 3.0\}$ (20 configurations), with 3 histogram
bins and $B = 100$ null replicates per configuration.  When all
communities are included, $z$-scores range from 3.0 to 27.7, and all 20
configurations reject.  When singletons are excluded, all 20
configurations produce negative $z$-scores (range $-3.54$ to $-0.39$),
and none rejects.

For the NKI breast cancer dataset, we test resolution $\in \{15, 20, 25,
30\}$ and gain $\in \{2.0, 2.5, 3.0\}$ (12 configurations).  The
$z$-scores range from $-0.19$ to $3.34$, with 11 of 12 below the
critical value.  One configuration (resolution~20, gain~3.0) produces $z = 3.34$.
Because the grid was explored as a sensitivity analysis rather than a
pre-specified confirmatory test, we do not interpret this isolated
rejection as confirmatory.  It does, however, identify a parameter
setting where the conclusion differs.

We did not vary the number of histogram bins; the bins control the
granularity of within-preimage clustering and are less influential than
resolution and gain.

\section{Detailed empirical results}\label{sm:empirical}

We report two additional diagnostic analyses for the empirical case
studies.  The first replaces the block-averaged dissociation statistic
with a single-feature maximum; the second examines the modularity
distributions produced by the null model.

\subsection{$D_{\max}$ analysis}

As a sensitivity check, we compute the largest absolute difference in
community means across all features and community pairs:
\[
  D_{\max} = \max_{k < \ell} \max_j
  |\bar{x}_{kj} - \bar{x}_{\ell j}|.
\]
This statistic avoids the block-averaging step and is sensitive to
single-feature differences.  We use the same observed and null Mapper
graphs; only the graph-level test statistic changes.

For the NKI breast cancer 1D analysis: observed $D_{\max} = 2.44$, null
mean 1.66 (sd 1.47), $z = 0.53$, $\hat{p} = 0.29$.  The $D_{\max}$
result agrees with the block-averaged dissociation.

For the NBA analysis under Configuration~1: $z = 18.4$ with all
communities, $z = -3.6$ without singletons.  Under Configuration~2:
$z = 15.1$ with all communities, $z = -1.5$ without.  The
singleton-driven pattern is robust to the choice of test statistic.

\subsection{Null modularity distributions}

The null modularity distribution provides context about the graph-level
structure produced under the null.

For NKI 1D: the observed graph has modularity $Q = 0.648$; the null
replicates have mean modularity $0.655$ (sd $0.021$, 95\% interval
$[0.604, 0.687]$) and produce $5.6$ communities on average.  Both the
observed modularity and the number of communities fall within the null
range.

For the NBA under Configuration~1: the null Mapper graphs have mean
modularity $0.938$ (sd $0.007$) and produce $103$ communities on average,
compared to the observed $Q = 0.707$ with 40 communities.  Under
Configuration~2: null mean $Q = 0.732$ (sd $0.022$) with $8.5$
communities versus the observed $Q = 0.452$ with 7.  In both cases, the
null graphs are more modular than the observed graph, because Louvain
partitions the diffuse null data into many small, tightly separated
communities.

\section{Scope of replication}\label{sm:replication-scope}

Our replication of~\cite{Lum2013} is approximate in several respects.
The original analyses were conducted in Ayasdi, whose internal
parameterization differs from open-source Mapper implementations
(\cref{sm:ayasdi}).  We matched the graph structure as closely as
possible through parameter sweeps and equalized binning, but we cannot
guarantee that our Mapper graphs are identical to those of
\cite{Lum2013}.

The NKI dataset we use contains 336 samples rather than the 272 reported
in~\cite{Lum2013}, likely due to different filtering criteria that we
were unable to determine from the published methods.  These discrepancies
introduce uncertainty about whether our replication faithfully represents
the original analysis.  However, the non-rejection finding is consistent
across both the 1D and 2D configurations, with $z$-scores near zero in
both cases.  The feature-split robustness analysis
(\cref{sm:feature-splits}), in which all 51 splits produce
non-rejection, further suggests that the conclusion is stable.

For the NBA analysis, the replication is more straightforward: the
7-dimensional data are simple, the sample size matches exactly
($n = 452$), and the Mapper parameters translated well.  Our open-source
equalized Mapper produces 14 multi-player communities (40 total,
including 26 singletons) under the paper-matching configuration,
whereas~\cite{Lum2013} report 13 total clusters in Ayasdi; the
difference reflects proprietary implementation details.

The LGG analysis (\citealp{RabadanCamara2020}) serves a distinct role:
rather than replicating~\cite{Lum2013}, it provides an external
validation of the structured null framework in a methodologically
different setting.  Our voting analysis uses a single Congress (2008);
the temporal dynamics documented in~\cite{Lum2013}---particularly the
shrinking of the moderate ``Central Group'' over time---might produce
different results under the structured null in earlier, more bipartisan
eras.

\section{Statistical validation literature}\label{sm:validation}

Within TDA, hypothesis testing has developed primarily for persistent
homology, a technique that tracks the appearance and disappearance of
topological features (connected components, loops, voids) as one sweeps
through a sequence of increasingly coarse approximations to the data.
\cite{Bobrowski2023} provided evidence for a universality law,
suggesting that the distribution of persistence values---the ratio of the
scale at which a feature disappears to the scale at which it
appears---in standard persistence diagrams of random point clouds may
converge to a left-skewed Gumbel distribution, independent of the
underlying data-generating distribution.
This suggests a parameter-free null for assessing individual persistent
homology features.  Complementary approaches include confidence regions
via the bottleneck bootstrap \citep{CarriereMichelOudot2018} and
permutation-based tests for persistence diagrams; see~\cite{Chazal2021}
for a survey.

\section{Ridge regularization details}\label{supp:ridge}

For analyses where ridge regularization is used (here, only the LGG
analysis; see the case-study descriptions above), we replace
$\hat{\Sigma}$ with $\hat{\Sigma} + \epsilon I_p$.  We compute the
eigenvalues of $\hat{\Sigma}$ and check the minimum eigenvalue
$\lambda_{\min}$.  When $\lambda_{\min}$ falls below $10^{-10}$, we
set
\begin{equation}\label{eq:ridge-epsilon}
  \epsilon = \max(0, -\lambda_{\min}) + 10^{-6},
\end{equation}
which shifts all eigenvalues upward and ensures that the smallest
eigenvalue of the regularized matrix is at least $10^{-6}$.  The
$\max(0, -\lambda_{\min})$ term accounts for the possibility that
floating-point arithmetic produces a small negative eigenvalue; in
exact arithmetic, covariance matrices are positive semidefinite and all
eigenvalues are non-negative.

\section{Proofs of theoretical results}\label{supp:proofs}

We restate each result from the main text and provide its proof.
Theorem~3.1 establishes that covariance structure alone produces
nonzero dissociation under a linear filter; Theorem~3.2 shows that
label permutation targets a different null hypothesis whose population
dissociation is zero.

\subsection{Proof of Theorem~3.1 (Covariance-driven dissociation)}

\settheoremnum{3.1}
\begin{theorem}
  Let $X \sim \mathcal{N}_p(0, \Sigma)$ with $\Sigma$ positive definite
  and $\lambda_1 > \lambda_2$, so that the unit leading eigenvector~$u_1$ is
  uniquely defined up to sign.  Let $f(x) = u_1^\top x$ and let $A, B$
  partition $\{1, \ldots, p\}$ into two nonempty blocks.
  Let $\{I_1,\ldots,I_K\}$ be a finite interval partition of~$\R$
  with $K \ge 2$ and each interval having positive probability
  under~$f$.  Define the population block mean
  $m_{k,A} = \frac{1}{|A|}\sum_{j \in A} E[X_j \mid f \in I_k]$
  and similarly $m_{k,B}$, and set
  \[
    D_{\mathrm{pop}} = \max_{k < \ell}
    \max\!\bigl(|m_{k,A} - m_{\ell,A}|,\;
    |m_{k,B} - m_{\ell,B}|\bigr).
  \]
  If $\max(|\bar{u}_{1,A}|, |\bar{u}_{1,B}|) > 0$, then
  \[
    D_{\mathrm{pop}} \;\ge\; \max\!\bigl(|\bar{u}_{1,A}|,\;
    |\bar{u}_{1,B}|\bigr) \cdot
    \bigl|E[f \mid f \in I_{\max}] - E[f \mid f \in I_{\min}]\bigr|
    \;>\; 0,
  \]
  where $I_{\min}$ and $I_{\max}$ denote the leftmost and rightmost
  intervals with positive probability.
\end{theorem}

\begin{proof}
  The argument proceeds in three stages.  First
  (Steps~1--3), we use properties of the multivariate Gaussian to show
  that each block mean $m_{k,A}$ factors into a \emph{loading term}
  $\bar{u}_{1,A}$ (determined by the covariance structure) times a
  \emph{filter term} $E[f \mid f \in I_k]$ (determined by which interval
  we condition on).  Second (Steps~4--5), we show that the filter terms
  are strictly monotone across intervals, so block means in different
  intervals must differ.  Third (Steps~6--7), we combine these facts to
  obtain the lower bound and conclude $D_{\mathrm{pop}} > 0$.

  \emph{Step~1: Conditional expectation of a Gaussian given a linear
  projection.}
  When a random vector $X$ is jointly Gaussian with a scalar
  $f = u_1^\top X$, the conditional expectation of~$X$ given $f = c$
  is the linear regression of~$X$ on~$f$:
  \[
    E[X \mid f = c]
    = \frac{\operatorname{Cov}(X, f)}{\operatorname{Var}(f)} \cdot c.
  \]
  This is a standard property of the multivariate Gaussian: for any
  jointly Gaussian $(X, f)$ with $E[X] = 0$ and $E[f] = 0$, the
  conditional expectation is linear in~$f$, with slope given by the
  covariance of~$X$ with~$f$ divided by the variance of~$f$.

  We now evaluate the covariance and variance.  Since $f = u_1^\top X$
  and $X$ has covariance~$\Sigma$,
  \[
    \operatorname{Cov}(X, f)
    = \operatorname{Cov}(X,\, u_1^\top X)
    = \Sigma\, u_1,
  \]
  and
  \[
    \operatorname{Var}(f)
    = \operatorname{Var}(u_1^\top X)
    = u_1^\top \Sigma\, u_1.
  \]
  Because $u_1$ is the leading eigenvector of~$\Sigma$ with
  eigenvalue~$\lambda_1$, we have $\Sigma\, u_1 = \lambda_1 u_1$ and
  $u_1^\top \Sigma\, u_1 = \lambda_1 u_1^\top u_1 = \lambda_1$ (since
  $u_1$ is a unit vector).  Substituting:
  \[
    E[X \mid f = c]
    = \frac{\lambda_1 u_1}{\lambda_1} \cdot c
    = u_1\, c.
  \]
  In particular, for each coordinate~$j$,
  \begin{equation}\label{eq:cond-exp-coord}
    E[X_j \mid f = c] = (u_1)_j \cdot c.
  \end{equation}

  \emph{Step~2: From pointwise conditioning to interval conditioning.}
  Step~1 gives $E[X_j \mid f = c]$ for a specific value~$c$, but the
  theorem conditions on $f \in I_k$ (an interval) rather than
  $f = c$ (a point).  We pass between the two using the tower property
  of conditional expectation:
  \[
    E[X_j \mid f \in I_k]
    = E\bigl[\,E[X_j \mid f]\;\big|\; f \in I_k\bigr].
  \]
  The inner expectation is $E[X_j \mid f = c] = (u_1)_j \cdot c$
  from~\eqref{eq:cond-exp-coord}.  Substituting and using linearity
  of expectation:
  \[
    E[X_j \mid f \in I_k]
    = E\bigl[(u_1)_j \cdot f \;\big|\; f \in I_k\bigr]
    = (u_1)_j \cdot E[f \mid f \in I_k].
  \]

  \emph{Step~3: Block averaging.}
  We now connect the coordinatewise result from Step~2 to the
  block-level quantity $m_{k,A}$ that appears in the dissociation
  metric.  The population block mean for community~$k$ and block~$A$ is
  $m_{k,A} = \frac{1}{|A|}\sum_{j \in A} E[X_j \mid f \in I_k]$.
  Substituting the result of Step~2:
  \[
    m_{k,A}
    = \frac{1}{|A|}\sum_{j \in A} (u_1)_j \cdot E[f \mid f \in I_k]
    = \underbrace{\biggl(\frac{1}{|A|}\sum_{j \in A}
      (u_1)_j\biggr)}_{\bar{u}_{1,A}}
      \cdot\; E[f \mid f \in I_k],
  \]
  where the factor $E[f \mid f \in I_k]$ does not depend on~$j$ and
  can be pulled out of the sum.  The same identity holds for block~$B$:
  $m_{k,B} = \bar{u}_{1,B} \cdot E[f \mid f \in I_k]$.

  \emph{Step~4: Difference of block means between two intervals.}
  Steps~1--3 established that each block mean is a product of a loading
  and a filter term.  We now take the difference between two intervals
  to see how dissociation arises.  For any two intervals $I_k$
  and $I_\ell$,
  \begin{align}
    m_{k,A} - m_{\ell,A}
    &= \bar{u}_{1,A} \cdot E[f \mid f \in I_k]
       - \bar{u}_{1,A} \cdot E[f \mid f \in I_\ell] \notag \\
    &= \bar{u}_{1,A} \cdot
       \bigl(E[f \mid f \in I_k] - E[f \mid f \in I_\ell]\bigr).
       \label{eq:block-diff}
  \end{align}

  \emph{Step~5: Monotonicity of truncated Gaussian means.}
  Equation~\eqref{eq:block-diff} shows that block means differ between
  intervals whenever the filter terms $E[f \mid f \in I_k]$ differ.
  We now show that they always do.
  The filter value satisfies $f \sim N(0, \lambda_1)$.  If $I_k$ lies
  entirely to the left of~$I_\ell$ on the real line, then conditioning
  $f$ on the left interval pulls the mean down and conditioning on the
  right interval pushes it up.  More precisely, for any Gaussian random
  variable and any two disjoint intervals $I_k < I_\ell$ (meaning every
  point of~$I_k$ is less than every point of~$I_\ell$), the truncated
  means satisfy
  \[
    E[f \mid f \in I_k] \;<\; E[f \mid f \in I_\ell].
  \]
  This follows from the fact that the Gaussian density is positive
  everywhere on~$\R$: conditioning on a left interval concentrates
  probability mass on smaller values, strictly decreasing the mean
  relative to conditioning on a right interval.

  \emph{Step~6: Lower bound via the extreme intervals.}
  We have all the ingredients: block means factor as loading times
  filter term (Step~3), and filter terms are strictly monotone
  (Step~5).  It remains to assemble these into a bound
  on~$D_{\mathrm{pop}}$.
  The dissociation $D_{\mathrm{pop}}$ takes a maximum over all pairs
  $k < \ell$ and both blocks.  We obtain a lower bound by restricting
  to the single pair $(I_{\min}, I_{\max})$---the leftmost and
  rightmost intervals with positive probability---and to block~$A$
  alone:
  \[
    D_{\mathrm{pop}}
    = \max_{k < \ell}\,\max\!\bigl(|m_{k,A} - m_{\ell,A}|,\;
      |m_{k,B} - m_{\ell,B}|\bigr)
    \;\ge\; |m_{\min,A} - m_{\max,A}|.
  \]
  Applying~\eqref{eq:block-diff} with $k = \min$ and $\ell = \max$
  and taking absolute values:
  \[
    |m_{\min,A} - m_{\max,A}|
    = |\bar{u}_{1,A}| \cdot
      \bigl|E[f \mid f \in I_{\max}] - E[f \mid f \in I_{\min}]\bigr|.
  \]
  The same argument with block~$B$ gives
  $D_{\mathrm{pop}} \ge |\bar{u}_{1,B}| \cdot
  |E[f \mid f \in I_{\max}] - E[f \mid f \in I_{\min}]|$.
  Since $D_{\mathrm{pop}}$ is at least as large as both the block-$A$
  and block-$B$ contributions, it is at least as large as the larger
  of the two:
  \[
    D_{\mathrm{pop}}
    \;\ge\; \max\!\bigl(|\bar{u}_{1,A}|,\; |\bar{u}_{1,B}|\bigr)
    \cdot \bigl|E[f \mid f \in I_{\max}]
    - E[f \mid f \in I_{\min}]\bigr|.
  \]

  \emph{Step~7: Strict positivity.}
  The first factor, $\max(|\bar{u}_{1,A}|, |\bar{u}_{1,B}|)$, is
  positive by the loading condition assumed in the theorem.  The second
  factor, $|E[f \mid f \in I_{\max}] - E[f \mid f \in I_{\min}]|$, is
  positive by the monotonicity established in Step~5, since $I_{\min}$
  and $I_{\max}$ are distinct intervals with $I_{\min} < I_{\max}$.
  Therefore $D_{\mathrm{pop}} > 0$.
\end{proof}
\restoretheoremnum

\subsection{Genericity of the loading condition}\label{supp:genericity}

\begin{proposition}\label{prop:genericity}
  For any partition of $\{1, \ldots, p\}$ into nonempty blocks $A$ and
  $B$, the set of positive definite matrices $\Sigma$ for which
  $\bar{u}_{1,A} = \bar{u}_{1,B} = 0$ has Lebesgue measure zero.
\end{proposition}

\begin{proof}
  The proof strategy is as follows.  We want to show that for ``almost
  every'' positive definite matrix~$\Sigma$, the leading eigenvector
  $u_1$ does not simultaneously average to zero on both blocks $A$
  and~$B$.  We do this in five steps.  First, we identify the ``bad'' directions on
  the unit sphere where both averages vanish (Steps~1--2) and show that
  these directions form a negligibly small set (Step~2).  Next, we remove
  a negligible set of matrices with repeated eigenvalues so that $u_1$ is
  well-defined (Step~3).  Finally, we use the spectral decomposition to
  show that the leading eigenvector is essentially uniformly distributed
  on the sphere, so it avoids the bad set with probability one
  (Steps~4--5).

  \emph{Step~1: Rewriting the block averages as inner products.}
  Let $\mathbf{1}_A \in \R^p$ be the indicator vector of block~$A$:
  its $j$th entry is~1 if $j \in A$ and~0 otherwise; define
  $\mathbf{1}_B$ similarly for block~$B$.  Set
  $a = \frac{1}{|A|}\mathbf{1}_A$ and
  $b = \frac{1}{|B|}\mathbf{1}_B$.  Then the block averages of the
  leading eigenvector can be written as inner products:
  \[
    \bar{u}_{1,A}
    = \frac{1}{|A|}\sum_{j \in A} (u_1)_j
    = \sum_{j=1}^p a_j\, (u_1)_j
    = \langle u_1, a \rangle,
  \]
  and similarly $\bar{u}_{1,B} = \langle u_1, b \rangle$.

  The vectors $a$ and $b$ are linearly independent: $a$ has positive
  entries exactly on~$A$ and zeros on~$B$, while $b$ has positive
  entries exactly on~$B$ and zeros on~$A$.  Since $A$ and $B$ are
  nonempty and disjoint, neither vector is a scalar multiple of the
  other.

  \emph{Step~2: The bad set on the sphere is small.}
  Let $S^{p-1} = \{v \in \R^p : \|v\| = 1\}$ denote the unit sphere
  in~$\R^p$.  Since $u_1$ is a unit eigenvector, it lies on~$S^{p-1}$.
  Define the ``bad set''
  \[
    S_{\mathrm{bad}}
    = \{ v \in S^{p-1} : \langle v, a \rangle = 0
    \text{ and } \langle v, b \rangle = 0 \}.
  \]
  Each condition $\langle v, a \rangle = 0$ restricts $v$ to a
  hyperplane through the origin (the set of vectors orthogonal to~$a$).
  Intersecting $S^{p-1}$ with one hyperplane gives a
  $(p-2)$-dimensional subsphere.  Imposing the second independent
  constraint $\langle v, b \rangle = 0$ cuts the dimension by one more.
  Concretely:

  If $p = 2$, the two independent constraints $\langle v, a \rangle = 0$
  and $\langle v, b \rangle = 0$ leave only $v = 0$, which is not on
  the unit sphere, so $S_{\mathrm{bad}} = \varnothing$.

  If $p \ge 3$, $S_{\mathrm{bad}}$ is the intersection of~$S^{p-1}$
  with a $(p - 2)$-dimensional linear subspace, yielding a smooth
  submanifold of dimension $p - 3$.  A $(p-3)$-dimensional submanifold
  of the $(p-1)$-dimensional sphere has surface area zero---just as a
  curve (1-dimensional) has zero area on a surface (2-dimensional).

  In both cases, $S_{\mathrm{bad}}$ has measure zero under the uniform
  (surface area) measure on~$S^{p-1}$.

  \emph{Step~3: Removing matrices with repeated eigenvalues.}
  When $\Sigma$ has a repeated largest eigenvalue
  ($\lambda_1 = \lambda_2$), the leading eigenvector $u_1$ is not
  uniquely defined (any unit vector in the eigenspace would work), so
  we cannot meaningfully ask whether $u_1$ avoids $S_{\mathrm{bad}}$.
  We handle this by showing that such matrices are negligibly rare.

  The eigenvalues of a symmetric matrix are the roots of its
  characteristic polynomial $\det(\Sigma - \lambda I)$.  Two
  eigenvalues coincide exactly when this polynomial has a repeated
  root, which happens when its discriminant---a fixed nonzero polynomial
  in the entries of~$\Sigma$---equals zero.  The zero set of a nonzero
  polynomial in~$\R^N$ has Lebesgue measure zero (this is a standard
  fact from real algebraic geometry).  Therefore the set of positive
  definite matrices with any repeated eigenvalue has Lebesgue measure
  zero.

  It suffices to prove the proposition on the remaining matrices, which
  we call the simple-spectrum set:
  $\mathrm{SPD}_p^\circ = \{ \Sigma \in \mathrm{SPD}_p :
  \lambda_1(\Sigma) > \lambda_2(\Sigma) > \cdots > \lambda_p(\Sigma)
  > 0 \}$.  On this set, $u_1$ is uniquely defined up to sign, which
  does not affect $\bar{u}_{1,A}$ or~$\bar{u}_{1,B}$ being zero.

  \emph{Step~4: Spectral-decomposition parametrization.}
  Every matrix in $\mathrm{SPD}_p^\circ$ can be written as
  $\Sigma = Q\, \mathrm{diag}(\lambda_1, \ldots, \lambda_p)\, Q^\top$,
  where $Q$ is a $p \times p$ orthogonal matrix (i.e., $Q^\top Q = I$)
  and $\lambda_1 > \cdots > \lambda_p > 0$ are the ordered eigenvalues.
  This is the spectral theorem from linear algebra.  The columns of~$Q$
  are the eigenvectors: the first column $q_1$ is precisely the leading
  eigenvector~$u_1$.

  The idea of the remaining argument is that specifying~$\Sigma$ is
  equivalent to specifying its eigenvalues and its eigenvectors
  separately.  If we can show that the eigenvector part is ``uniformly
  spread'' over all directions, then hitting the measure-zero bad set
  $S_{\mathrm{bad}}$ from Step~2 has probability zero.  Making this
  precise requires a result about how Lebesgue measure on symmetric
  matrices decomposes into an eigenvalue part and an eigenvector part.
  This result is standard in random matrix theory; we state it here
  for completeness and provide a reference.

  \begin{fact}\label{fact:spectral-measure}
    Restricted to the simple-spectrum set $\mathrm{SPD}_p^\circ$,
    Lebesgue measure on the space of $p \times p$ symmetric matrices
    decomposes as
    \[
      d\Sigma \;\propto\;
      \prod_{i < j} |\lambda_i - \lambda_j|\;
      d\lambda_1 \cdots d\lambda_p \; dQ,
    \]
    where $d\lambda_1 \cdots d\lambda_p$ is Lebesgue measure on the
    eigenvalues and $dQ$ is Haar measure on the orthogonal
    group~$O(p)$.
  \end{fact}

  We briefly unpack the two components of this decomposition for
  readers unfamiliar with them:

  \emph{Haar measure $dQ$ on~$O(p)$} is the unique probability measure
  on the orthogonal group that is invariant under orthogonal
  transformations: if $Q$ is drawn from Haar measure and $R$ is any
  fixed orthogonal matrix, then $RQ$ has the same distribution
  as~$Q$.  It is the natural analog of the uniform distribution, but
  for the orthogonal group rather than for an interval.

  \emph{The Vandermonde factor} $\prod_{i<j}|\lambda_i - \lambda_j|$
  is the absolute value of the determinant of the Jacobian matrix for
  the change of variables from matrix entries to eigenvalues and
  eigenvectors.  Its exact form is not needed for our argument; what
  matters is that it is a scalar that depends only on the eigenvalues,
  not on~$Q$.

  Because the Vandermonde factor depends only on the eigenvalues, the
  decomposition in Fact~\ref{fact:spectral-measure} tells us that the
  measure factorizes into an eigenvalue-dependent density times Haar
  measure on~$O(p)$.  For a proof of this decomposition, see
  \cite[Theorem~2.6.1]{AndersonGuionnetZeitouni2010}.

  \emph{Step~5: The leading eigenvector is uniform on the sphere.}
  Under Haar measure on~$O(p)$, the first column~$q_1$ of a random
  orthogonal matrix is uniformly distributed on the unit
  sphere~$S^{p-1}$.  To see why, let $R$ be any fixed rotation.  By
  the invariance property of Haar measure, $RQ$ has the same
  distribution as~$Q$, so the first column of~$RQ$---which is
  $Rq_1$---has the same distribution as~$q_1$.  The only distribution
  on~$S^{p-1}$ with this property (invariance under all rotations) is
  the uniform distribution.

  We can now complete the argument.  Fix any choice of eigenvalues
  $\lambda_1 > \cdots > \lambda_p > 0$.  The leading eigenvector is
  $u_1 = q_1$, and $q_1$ is uniform on~$S^{p-1}$.  Since
  $S_{\mathrm{bad}}$ has surface-area measure zero on~$S^{p-1}$
  (Step~2), the probability that $q_1$ lands in $S_{\mathrm{bad}}$ is
  zero:
  \[
    \Pr(q_1 \in S_{\mathrm{bad}}) = 0.
  \]
  This holds for every choice of eigenvalues.  Because the measure
  factorizes into eigenvalue and eigenvector parts
  (Fact~\ref{fact:spectral-measure}), integrating over all eigenvalues
  preserves this conclusion: the full set of simple-spectrum positive
  definite matrices with $\bar{u}_{1,A} = \bar{u}_{1,B} = 0$ has
  Lebesgue measure zero.  Adding back the measure-zero
  repeated-eigenvalue set from Step~3 does not change this conclusion.
\end{proof}

\subsection{Proof of Theorem~3.2 (Label-permutation null)}

\settheoremnum{3.2}
\begin{theorem}
  Suppose the number of communities~$K$ is fixed and the community
  sizes satisfy $n_k / n \to \pi_k \in (0,1)$ as $n \to \infty$ for
  $k = 1, \ldots, K$.  Conditional on any observed data for which the
  finite-population variances $S_A^2$ and $S_B^2$ remain bounded, under
  a uniformly random reassignment of data points to communities,
  preserving the sizes~$n_k$, the permuted dissociation satisfies
  converges to zero in probability at rate $n^{-1/2}$.
  That is, as the sample size grows, the probability that
  $D^{\mathrm{perm}}$ exceeds any fixed threshold vanishes.
\end{theorem}

\begin{proof}
  The argument proceeds in four steps.  First, we reduce the
  multivariate problem to a univariate one by introducing block
  averages.  Second, we observe that a random permutation of community
  labels turns each community into a simple random sample, and we
  derive the variance of the resulting community means using a
  finite-population sampling formula.  Third, we use this variance to
  show that every permuted community mean concentrates around the
  overall (grand) mean at rate $n^{-1/2}$.  Finally, we show that the
  dissociation---which measures differences between community
  means---inherits this rate.

  \emph{Step~1: Block averages.}
  For each data point $i = 1, \ldots, n$, define the block-$A$ average
  $Y_{i,A} = \frac{1}{|A|}\sum_{j \in A} X_{ij}$ and similarly
  $Y_{i,B} = \frac{1}{|B|}\sum_{j \in B} X_{ij}$.  These reduce the
  $p$-dimensional data to two scalar summaries per data point: the
  average of the features in block~$A$ and the average of the features
  in block~$B$.  The dissociation metric $D$ depends on the data only
  through these block averages (by definition), so it suffices to work
  with $Y_{i,A}$ and $Y_{i,B}$.

  Also define the grand block mean, which is the average of $Y_{i,A}$
  over all $n$ data points:
  \[
    \bar{Y}_{\bullet,A}
    = \frac{1}{n}\sum_{i=1}^n Y_{i,A}.
  \]

  \emph{Step~2: Permutation as simple random sampling.}
  Under the label-permutation null, we randomly reassign data points
  to communities while keeping the community sizes $n_1, \ldots, n_K$
  fixed.  Concretely, this means we draw a uniformly random partition
  of $\{1, \ldots, n\}$ into groups of sizes $n_1, \ldots, n_K$.

  From the perspective of any single community~$k$, its $n_k$ members
  are a simple random sample drawn without replacement from the finite
  population of all $n$ data points.  (This is exactly what ``uniformly
  random partition'' means: every subset of size~$n_k$ is equally
  likely to be assigned to community~$k$.)

  The mean of $Y_{i,A}$ within the permuted community~$k$ is
  \[
    \bar{Y}_{k,A}^{\mathrm{perm}}
    = \frac{1}{n_k}\sum_{i \in C_k^{\mathrm{perm}}} Y_{i,A},
  \]
  where $C_k^{\mathrm{perm}}$ denotes the set of data points assigned
  to community~$k$ under the permutation.  Because this is the mean of
  a simple random sample without replacement from a finite population,
  its variance (conditional on the observed data~$X$, which determines
  the values $Y_{1,A}, \ldots, Y_{n,A}$) is given by the
  finite-population sampling formula:
  \begin{equation}\label{eq:fpc-variance}
    \operatorname{Var}\!\bigl(\bar{Y}_{k,A}^{\mathrm{perm}}
    \mid X\bigr)
    = \frac{S_A^2}{n_k}\,
      \Bigl(1 - \frac{n_k}{n}\Bigr)
    = \Bigl(\frac{1}{n_k} - \frac{1}{n}\Bigr) S_A^2,
  \end{equation}
  where
  \[
    S_A^2
    = \frac{1}{n - 1}\sum_{i=1}^n
      \bigl(Y_{i,A} - \bar{Y}_{\bullet,A}\bigr)^2
  \]
  is the sample variance of the $Y_{i,A}$ values across all $n$ data
  points.  The factor $(1 - n_k/n)$ is the finite-population correction:
  sampling without replacement from a finite population produces less
  variability than sampling with replacement, because drawing a value
  removes it from the pool and makes future draws slightly more
  predictable.  When $n_k$ is small relative to~$n$ (so that each
  community is a small fraction of the total), this correction is close
  to~1 and the formula resembles the familiar $S_A^2 / n_k$.

  \emph{Step~3: Each permuted community mean concentrates around the
  grand mean.}
  We now use~\eqref{eq:fpc-variance} to show that
  $\bar{Y}_{k,A}^{\mathrm{perm}}$ is close to the grand mean
  $\bar{Y}_{\bullet,A}$.

  First, note that
  $E[\bar{Y}_{k,A}^{\mathrm{perm}} \mid X] = \bar{Y}_{\bullet,A}$:
  the expected value of a simple random sample mean equals the
  population mean.  So the difference
  $\bar{Y}_{k,A}^{\mathrm{perm}} - \bar{Y}_{\bullet,A}$ has
  conditional mean zero and conditional variance
  $(1/n_k - 1/n)\, S_A^2$.

  To bound the size of this difference, we use the assumption that
  $n_k / n \to \pi_k \in (0, 1)$ as $n \to \infty$.  This means $n_k$ grows
  proportionally to~$n$, so $(1/n_k - 1/n)$ is of order $1/n$.
  More precisely, since $n_k/n \to \pi_k$,
  \[
    n\Bigl(\frac{1}{n_k} - \frac{1}{n}\Bigr)
    = \frac{n}{n_k} - 1
    \;\to\; \frac{1}{\pi_k} - 1,
  \]
  so $(1/n_k - 1/n) = O(1/n)$.
  Therefore
  $\operatorname{Var}(\bar{Y}_{k,A}^{\mathrm{perm}} \mid X)$ is of
  order $S_A^2 / n$.  Because we condition on the observed data, $S_A^2$
  is a fixed constant (not random).  The theorem assumes
  $S_A^2, S_B^2 \le C_0$ for some bound~$C_0$; combined with the
  proportional-size assumption, this gives
  $\operatorname{Var}(\bar{Y}_{k,A}^{\mathrm{perm}} \mid X) \le
  C_1 / n$ for a constant $C_1$ depending only on $C_0$ and
  the~$\pi_k$.

  By Chebyshev's inequality, a random variable with mean zero and
  variance $\sigma^2$ satisfies
  $\Pr(|Z| > t) \le \sigma^2 / t^2$.  Applying this conditionally
  (given~$X$) with
  $\sigma^2 \le C_1/n$
  and $t = C\, n^{-1/2}$ for any constant $C > 0$ gives
  \[
    \Pr\!\bigl(|\bar{Y}_{k,A}^{\mathrm{perm}}
    - \bar{Y}_{\bullet,A}| > C\, n^{-1/2}
    \;\big|\; X\bigr)
    \;\le\; \frac{C_1 / n}{C^2 / n}
    \;=\; \frac{C_1}{C^2},
  \]
  which can be made arbitrarily small by choosing $C$ large.  Therefore
  \[
    \bar{Y}_{k,A}^{\mathrm{perm}} - \bar{Y}_{\bullet,A}
    = O_p(n^{-1/2}).
  \]
  The same argument applies to block~$B$:
  $\bar{Y}_{k,B}^{\mathrm{perm}} - \bar{Y}_{\bullet,B}
  = O_p(n^{-1/2})$.

  \emph{Step~4: From community means to dissociation.}
  The permuted dissociation is
  \[
    D^{\mathrm{perm}}
    = \max_{k < \ell}\,
      \max\!\bigl(
        |\bar{Y}_{k,A}^{\mathrm{perm}}
         - \bar{Y}_{\ell,A}^{\mathrm{perm}}|,\;
        |\bar{Y}_{k,B}^{\mathrm{perm}}
         - \bar{Y}_{\ell,B}^{\mathrm{perm}}|
      \bigr).
  \]
  Each difference of permuted community means can be written as
  \[
    \bar{Y}_{k,A}^{\mathrm{perm}} - \bar{Y}_{\ell,A}^{\mathrm{perm}}
    = \bigl(\bar{Y}_{k,A}^{\mathrm{perm}}
      - \bar{Y}_{\bullet,A}\bigr)
    - \bigl(\bar{Y}_{\ell,A}^{\mathrm{perm}}
      - \bar{Y}_{\bullet,A}\bigr).
  \]
  Each of the two terms in parentheses is $O_p(n^{-1/2})$ by Step~3,
  so their difference is also $O_p(n^{-1/2})$.  (The sum or difference
  of two quantities that are each $O_p(n^{-1/2})$ is itself
  $O_p(n^{-1/2})$, by the triangle inequality.)

  The dissociation $D^{\mathrm{perm}}$ takes a maximum over
  $\binom{K}{2}$ community pairs and over both blocks.  Since $K$ is
  fixed (it does not grow with~$n$), this is a maximum of a fixed,
  finite number of $O_p(n^{-1/2})$ terms.  The maximum of finitely many
  terms that are each $O_p(n^{-1/2})$ is itself $O_p(n^{-1/2})$,
  because the largest of (say) 10 small quantities is still small.
  Therefore
  \[
    D^{\mathrm{perm}} = O_p(n^{-1/2}).
  \]
  As $n \to \infty$, $D^{\mathrm{perm}} \to 0$ in probability.  The
  population target of the label-permutation null is therefore zero:
  every permuted community mean converges to the same grand mean, so
  the differences between communities vanish.
\end{proof}
\restoretheoremnum

\bibliographystyle{siamplain}
\bibliography{references}

@book{AndersonGuionnetZeitouni2010,
  author    = {Anderson, Greg W. and Guionnet, Alice and Zeitouni, Ofer},
  title     = {An Introduction to Random Matrices},
  publisher = {Cambridge University Press},
  year      = {2010}
}

@book{CoverThomas2006,
  author    = {Cover, Thomas M. and Thomas, Joy A.},
  title     = {Elements of Information Theory},
  edition   = {2nd},
  publisher = {Wiley-Interscience},
  year      = {2006}
}

@article{LedoitWolf2004,
  author  = {Ledoit, Olivier and Wolf, Michael},
  title   = {A well-conditioned estimator for large-dimensional covariance matrices},
  journal = {Journal of Multivariate Analysis},
  volume  = {88},
  number  = {2},
  pages   = {365--411},
  year    = {2004}
}

@article{Cancer2015lgg,
  author  = {{Cancer Genome Atlas Research Network}},
  title   = {Comprehensive, integrative genomic analysis of diffuse lower-grade gliomas},
  journal = {New England Journal of Medicine},
  volume  = {372},
  number  = {26},
  pages   = {2481--2498},
  year    = {2015}
}

@book{RabadanCamara2020,
  author    = {Rabadan, Raul and Camara, Pablo G.},
  title     = {Topological Data Analysis for Genomics and Evolution: Topology in Biology},
  publisher = {Cambridge University Press},
  year      = {2020}
}

@article{vanDeVijver2002,
  author  = {van de Vijver, Marc J. and He, Yudong D. and van't Veer, Laura J. and Dai, Hongyue and Hart, Augustinus A. M. and Voskuil, Dorien W. and Schreiber, George J. and Peterse, Johannes L. and Roberts, Chris and Marton, Matthew J. and others},
  title   = {A gene-expression signature as a predictor of survival in breast cancer},
  journal = {New England Journal of Medicine},
  volume  = {347},
  number  = {25},
  pages   = {1999--2009},
  year    = {2002}
}

@article{Blondel2008,
  author  = {Blondel, Vincent D. and Guillaume, Jean-Loup and Lambiotte, Renaud and Lefebvre, Etienne},
  title   = {Fast unfolding of communities in large networks},
  journal = {Journal of Statistical Mechanics: Theory and Experiment},
  volume  = {2008},
  number  = {10},
  pages   = {P10008},
  year    = {2008}
}

@inproceedings{Singh2007,
  author    = {Singh, Gurjeet and M{\'e}moli, Facundo and Carlsson, Gunnar},
  title     = {Topological Methods for the Analysis of High Dimensional Data Sets and {3D} Object Recognition},
  booktitle = {Eurographics Symposium on Point-Based Graphics},
  year      = {2007}
}

@article{CarriereOudot2018,
  author  = {Carri{\`e}re, Mathieu and Oudot, Steve},
  title   = {Structure and Stability of the One-Dimensional {Mapper}},
  journal = {Foundations of Computational Mathematics},
  volume  = {18},
  pages   = {1333--1396},
  year    = {2018}
}

@article{Brown2021,
  author  = {Brown, Adam and Bobrowski, Omer and Munch, Elizabeth and Wang, Bei},
  title   = {Probabilistic Convergence and Stability of Random {Mapper} Graphs},
  journal = {Journal of Applied and Computational Topology},
  volume  = {5},
  pages   = {45--96},
  year    = {2021}
}

@article{CarriereMichelOudot2018,
  author  = {Carri{\`e}re, Mathieu and Michel, Bertrand and Oudot, Steve},
  title   = {Statistical Analysis and Parameter Selection for {Mapper}},
  journal = {Journal of Machine Learning Research},
  volume  = {19},
  number  = {12},
  pages   = {1--40},
  year    = {2018}
}

@inproceedings{Dey2016,
  author    = {Dey, Tamal K. and M{\'e}moli, Facundo and Wang, Yusu},
  title     = {Multiscale {Mapper}: Topological Summarization via Codomain Covers},
  booktitle = {Proceedings of the 27th ACM-SIAM Symposium on Discrete Algorithms (SODA)},
  year      = {2016}
}

@inproceedings{Dey2017,
  author    = {Dey, Tamal K. and M{\'e}moli, Facundo and Wang, Yusu},
  title     = {Topological Analysis of Nerves, {Reeb} Spaces, {Mappers}, and Multiscale {Mappers}},
  booktitle = {33rd International Symposium on Computational Geometry (SoCG)},
  series    = {LIPIcs},
  volume    = {77},
  pages     = {36:1--36:16},
  year      = {2017}
}

@article{Alvarado2025,
  author  = {Alvarado, Enrique and Belton, Robin and Fischer, Emily and Lee, Kang-Ju and Palande, Sourabh and Percival, Sarah and Purvine, Emilie},
  title   = {{G-Mapper}: Learning a Cover in the {Mapper} Construction},
  journal = {SIAM Journal on Mathematics of Data Science},
  volume  = {7},
  number  = {1},
  pages   = {353--385},
  year    = {2025},
  doi     = {10.1137/24M1641312}
}

@article{Tao2025,
  author  = {Tao, Yuyang and Ge, Shufei},
  title   = {A Distribution-Guided {Mapper} Algorithm},
  journal = {BMC Bioinformatics},
  volume  = {26},
  pages   = {73},
  year    = {2025},
  doi     = {10.1186/s12859-025-06085-5}
}

@article{Tauzin2021,
  author  = {Tauzin, Guillaume and Lupo, Umberto and others},
  title   = {\texttt{giotto-tda}: A Topological Data Analysis Toolkit for Machine Learning and Data Exploration},
  journal = {Journal of Machine Learning Research},
  volume  = {22},
  year    = {2021}
}

@article{vanVeen2019,
  author  = {{van Veen}, Hendrik Jacob and Saul, Nathaniel and Eargle, David and Mangham, Sam W.},
  title   = {{Kepler Mapper}: A Flexible {Python} Implementation of the {Mapper} Algorithm},
  journal = {Journal of Open Source Software},
  volume  = {4},
  number  = {42},
  pages   = {1315},
  year    = {2019}
}

@article{VejdemoJohansson2025,
  author  = {Vejdemo-Johansson, Mikael},
  title   = {$k$-Means Considered Harmful: On Arbitrary Topological Changes in {Mapper} Complexes},
  journal = {arXiv preprint arXiv:2507.06212},
  year    = {2025}
}

@article{Alvarado2024,
  author  = {Alvarado, Enrique and Belton, Robin and Lee, Erin and Palande, Sourabh and Percival, Sarah and Purvine, Emilie and Tymochko, Sarah},
  title   = {Any Graph is a {Mapper} Graph},
  journal = {arXiv preprint arXiv:2408.11180},
  year    = {2024}
}

@article{Lum2013,
  author  = {Lum, P. Y. and Singh, G. and Lehman, A. and Ishkanov, T. and Vejdemo-Johansson, M. and Alagappan, M. and Carlsson, J. and Carlsson, G.},
  title   = {Extracting Insights from the Shape of Complex Data Using Topology},
  journal = {Scientific Reports},
  volume  = {3},
  pages   = {1236},
  year    = {2013}
}

@article{Nicolau2011,
  author  = {Nicolau, Monica and Levine, Arnold J. and Carlsson, Gunnar},
  title   = {Topology Based Data Analysis Identifies a Subgroup of Breast Cancers with a Unique Mutational Profile and Excellent Survival},
  journal = {Proceedings of the National Academy of Sciences},
  volume  = {108},
  number  = {17},
  pages   = {7265--7270},
  year    = {2011}
}

@article{Rafique2020,
  author  = {Rafique, Ovais and Mir, Amir H.},
  title   = {A Topological Approach for Cancer Subtyping from Gene Expression Data},
  journal = {Journal of Biomedical Informatics},
  volume  = {102},
  pages   = {103357},
  year    = {2020}
}

@article{Rizvi2017,
  author  = {Rizvi, Abbas H. and Camara, Pablo G. and Kandror, Elena K. and Roberts, Thomas J. and Schieren, Ira and Maniatis, Tom and Rabadan, Raul},
  title   = {Single-Cell Topological {RNA-Seq} Analysis Reveals Insights into Cellular Differentiation and Development},
  journal = {Nature Biotechnology},
  volume  = {35},
  number  = {6},
  pages   = {551--560},
  year    = {2017}
}

@article{Saggar2018,
  author  = {Saggar, Manish and Sporns, Olaf and Gonzalez-Castillo, Javier and Bandettini, Peter A. and Carlsson, Gunnar and Glover, Gary and Reiss, Allan L.},
  title   = {Towards a New Approach to Reveal Dynamical Organization of the Brain Using Topological Data Analysis},
  journal = {Nature Communications},
  volume  = {9},
  number  = {1},
  pages   = {1399},
  year    = {2018}
}

@article{Geniesse2022,
  author  = {Geniesse, Caleb and Chowdhury, Samir and Saggar, Manish},
  title   = {{NeuMapper}: A Scalable Computational Framework for Multiscale Exploration of the Brain's Dynamical Organization},
  journal = {Network Neuroscience},
  volume  = {6},
  number  = {2},
  pages   = {467--498},
  year    = {2022}
}

@article{Hasegan2024,
  author  = {Ha{\c{s}}egan, Daniel and Geniesse, Caleb and Chowdhury, Samir and Saggar, Manish},
  title   = {Deconstructing the {Mapper} Algorithm to Extract Richer Topological and Temporal Features from Functional Neuroimaging Data},
  journal = {Network Neuroscience},
  volume  = {8},
  number  = {4},
  pages   = {1355--1379},
  year    = {2024}
}

@article{Li2015,
  author  = {Li, Li and Cheng, Wei-Yi and Glicksberg, Benjamin S. and Gottesman, Omri and Tamler, Ronald and Chen, Rong and Bottinger, Erwin P. and Dudley, Joel T.},
  title   = {Identification of Type~2 Diabetes Subgroups Through Topological Analysis of Patient Similarity},
  journal = {Science Translational Medicine},
  volume  = {7},
  number  = {311},
  pages   = {311ra174},
  year    = {2015}
}

@article{Wamil2023,
  author  = {Wamil, Maryam and Hassaine, Abdelaali and Rao, Shishir and others},
  title   = {Stratification of Diabetes in the Context of Comorbidities, Using Representation Learning and Topological Data Analysis},
  journal = {Scientific Reports},
  volume  = {13},
  pages   = {11478},
  year    = {2023}
}

@article{Lum2012,
  author  = {Lum, P. Y. and Lehmann, A. and Singh, G. and Ishkhanov, T. and Carlsson, G. and Vejdemo-Johansson, M.},
  title   = {The Topology of Politics: Voting Connectivity in the {US} {House of Representatives}},
  journal = {Physica A},
  volume  = {391},
  number  = {19},
  pages   = {4540--4547},
  year    = {2012}
}

@article{Yao2009,
  author  = {Yao, Yuan and Sun, Jian and Huang, Xuhui and Bowman, Gregory R. and Singh, Gurjeet and Lesnick, Michael and Guibas, Leonidas J. and Pande, Vijay S. and Carlsson, Gunnar},
  title   = {Topological Methods for Exploring Low-Density States in Biomolecular Folding Pathways},
  journal = {Journal of Chemical Physics},
  volume  = {130},
  number  = {14},
  pages   = {144115},
  year    = {2009}
}

@article{Tibshirani2001,
  author  = {Tibshirani, Robert and Walther, Guenther and Hastie, Trevor},
  title   = {Estimating the Number of Clusters in a Data Set via the Gap Statistic},
  journal = {Journal of the Royal Statistical Society, Series B},
  volume  = {63},
  number  = {2},
  pages   = {411--423},
  year    = {2001}
}

@article{Liu2008,
  author  = {Liu, Yufeng and Hayes, David Neil and Nobel, Andrew and Marron, J. S.},
  title   = {Statistical Significance of Clustering for High-Dimension, Low-Sample Size Data},
  journal = {Journal of the American Statistical Association},
  volume  = {103},
  number  = {483},
  pages   = {1281--1293},
  year    = {2008}
}

@article{Hartigan1985,
  author  = {Hartigan, J. A. and Hartigan, P. M.},
  title   = {The Dip Test of Unimodality},
  journal = {Annals of Statistics},
  volume  = {13},
  number  = {1},
  pages   = {70--84},
  year    = {1985}
}

@article{Fortunato2007,
  author  = {Fortunato, Santo and Barth{\'e}lemy, Marc},
  title   = {Resolution Limit in Community Detection},
  journal = {Proceedings of the National Academy of Sciences},
  volume  = {104},
  number  = {1},
  pages   = {36--41},
  year    = {2007}
}

@article{Silverman1981,
  author  = {Silverman, B. W.},
  title   = {Using Kernel Density Estimates to Investigate Multimodality},
  journal = {Journal of the Royal Statistical Society. Series B (Methodological)},
  volume  = {43},
  number  = {1},
  pages   = {97--99},
  year    = {1981}
}

@article{Chazal2021,
  author  = {Chazal, Fr{\'e}d{\'e}ric and Michel, Bertrand},
  title   = {An Introduction to Topological Data Analysis: Fundamental and Practical Aspects for Data Scientists},
  journal = {Frontiers in Artificial Intelligence},
  volume  = {4},
  pages   = {667963},
  year    = {2021},
  doi     = {10.3389/frai.2021.667963}
}

@article{Bobrowski2023,
  author  = {Bobrowski, Omer and Skraba, Primo{\v{z}}},
  title   = {A Universal Null-Distribution for Topological Data Analysis},
  journal = {Scientific Reports},
  volume  = {13},
  pages   = {12274},
  year    = {2023},
  doi     = {10.1038/s41598-023-37842-2}
}

@book{Davison1997,
  author    = {Davison, A. C. and Hinkley, D. V.},
  title     = {Bootstrap Methods and Their Application},
  publisher = {Cambridge University Press},
  year      = {1997},
  doi       = {10.1017/CBO9780511802843}
}

\end{document}